\documentclass[11pt]{article}
\usepackage{enumerate}
\usepackage[OT1]{fontenc}
\usepackage{amsmath,amssymb, amsthm}
\usepackage[square]{natbib}
\setcitestyle{authoryear}
\let\cite\citep
\usepackage[usenames]{color}
\usepackage{dsfont}
\usepackage{pgfplots}
\newtheorem{theorem}{Theorem}

\newtheorem{lemma}{Lemma} 
 
\newtheorem{definition}{Definition}
\newtheorem{remark}{Remark}
\newtheorem{assumption}{Assumption}
\newtheoremstyle{nonindented}{1ex}{1ex}{}{}{\bfseries}{.}{.5em}{}
\newtheoremstyle{indented}{1ex}{1ex}{\itshape\addtolength{\leftskip}{0.6cm}\addtolength{\rightskip}{0.6cm}}{}{\bfseries}{.}{.5em}{}
\theoremstyle{nonindented}
\theoremstyle{indented}
\theoremstyle{plain}

\newcommand{\abs}[1]{\left| #1 \right|}

\renewcommand{\hat}{\widehat}
\renewcommand{\tilde}{\widetilde}
\renewcommand{\bar}{\overline}

\DeclareMathOperator{\poly}{poly}

\def\max{\qopname\relax n{max}}

\def\argmax{\qopname\relax n{argmax}}

\def\Pr{\qopname\relax n{\mathbf{Pr}}}
\def\Ex{\qopname\relax n{\mathbf{E}}}
\def\supp{\qopname\relax n{\mathbf{supp}}}

\newcommand{\eat}[1]{}

\newcommand{\maxi}[1]{\mbox{maximize} & {#1 } & \\}

\newcommand{\st}{\mbox{subject to} }

\newcommand{\con}[1]{&#1 & \\}
\newcommand{\qcon}[2]{\displaystyle &#1 & \mbox{for } #2  \\}
\newenvironment{lp}{\begin{equation}  \begin{array}{l>{\displaystyle}l>{\displaystyle}l}}{\end{array}\end{equation}}
\newenvironment{lp*}{\begin{equation*}  \begin{array}{l>{\displaystyle}l>{\displaystyle}l}}{\end{array}\end{equation*}}

\usepackage{smile}
\usepackage{makecell}
\usepackage{multirow}
\usepackage{rotating}
\usepackage{enumerate}
\usepackage{esvect}
\usepackage{multirow}

\usepackage{tikz}
\usetikzlibrary{patterns}
\usetikzlibrary{arrows}

\usepackage{times}
\usepackage{setspace}
\usepackage{amsfonts}       %
\usepackage{nicefrac}       %
\usepackage{microtype}      %
\usepackage{mathtools}
\usepackage[ruled, vlined, linesnumbered]{algorithm2e}
\usepackage{fullpage}

\usepackage{float}

\usepackage[colorlinks,
            linkcolor=red,
            anchorcolor=blue,
            citecolor=blue
            ]{hyperref}

\usepackage{algorithmic}
\definecolor{green}{rgb}{0.09, 0.45, 0.27}

\title{Generalized Principal-Agency: Contracts,  Information, Games and Beyond}
\author{
Jiarui Gan \\ University of Oxford \\ \small jiarui.gan@cs.ox.ac.uk
\and
Minbiao Han \\ University of Chicago \\ \small minbiaohan@uchicago.edu
\and
Jibang Wu \\ University of Chicago \\ \small wujibang@uchicago.edu
\and
Haifeng Xu \\ University of Chicago \\ \small haifengxu@uchicago.edu
}

\begin{document}
\maketitle

\begin{abstract}
\noindent
In the principal-agent problem formulated by \citet{myerson1982optimal}, agents have private information (type) and make private decisions (action), both of which are unobservable to the principal. 
Myerson pointed out an elegant linear programming solution based on the revelation principle.   
This work extends Myerson's results to a more general setting where the principal's strategy space can be an arbitrary convex set, and thereby accommodating infinite strategy space and additional design constraints.   

Our generalized principal-agent model unifies several important design problems including contract design, information design, and Bayesian Stackelberg games, and encompasses them as special cases. 
We first extend the revelation principle to this general model, based on which a polynomial-time algorithm is then derived for computing the optimal   mechanism for the principal.  
This algorithm not only implies new efficient solutions simultaneously for all the aforementioned special cases  but also significantly  simplifies previously known algorithms designed for special cases. Inspired by the recent interest in the algorithmic design of a single contract and  menu of contracts,  we  study such constrained design problems to our general principal-agent model. In contrast to the above \emph{unification},  our results here illustrate the other facet of \emph{diversity} among different principal-agent design problems and demonstrate how their different structures   can lead to   different  complexities: some are tractable whereas others are APX-hard. Finally, we reveal an interesting  connection of our model to the problem of information acquisition for decision making and study its algorithmic properties in general.           

\end{abstract}

\section{Introduction}
The generalized principal-agency encapsulates two important economic concepts, \emph{adverse selection} and \emph{moral hazard}, both of which have been extensively analyzed and are foundational to seminal works in the field \cite{akerlof1978market, grossman1992analysis, antle1981moral, einav2013selection}.\footnote{These conceptual breakthroughs have been recognized as key contributions in the Nobel Prize in Economic Sciences awarded in 2001 and 2016.}  Adverse selection refers to the scenario where the agent holds private information not accessible to the principal, creating an information asymmetry. Moral hazard, on the other hand, occurs when the agent's actions cannot be directly monitored or enforced by the principal, which may lead to opportunistic behavior. While these factors together enable the generalized principal-agent problem to address a wide array of economic issues, they also pose substantial obstacles in designing effective strategies for the principal to interact with agents. 

The seminal work of \citet{myerson1982optimal} introduced the general principal-agent problem under an extended revelation principle~\cite{myerson1981optimal} and outlined an elegant linear programming approach for solving the principal's optimal coordination mechanism. Despite its significance in the economics literature,  this general and cleanly formulated model has not been studied from the algorithmic perspective. In this paper, through the algorithmic lens, we extend \citet{myerson1982optimal}'s mechanism design insights to a significantly broader class of principal-agent problems. Core to our development is the insight that the finite number of principal actions in \citet{myerson1982optimal}'s model is unnecessary. Specifically, we generalize Myerson's model to a more general framework that can accommodate convex principal strategy spaces (as opposed to discrete actions as in his original model) and additional design constraints.
Our general model encompasses as special cases several popular design problems such as contract design, information design, information pricing, and Bayesian Stackelberg games (see our summary in Section~\ref{sec:special-cases}). Moreover, we observe that the various (and seemingly irrelevant) solution concepts studied previously in these different design problems are merely instantiations of two natural restrictions to the principal's design space. 
This paper is set to connect the dots on the intrinsic structures of different principal-agent problems and provide a generic algorithmic approach to analyze their general properties under a unified framework. The elevated level of abstraction enables a cleaner and deeper understanding of the multifaceted nature of these design problems. Consequently, we are able to simplify existing proofs and address the remaining gaps in the literature, thereby painting a complete picture regarding these problems' algorithmic landscape.  

In addition to the design of coordination mechanisms,  another common practice for mitigating adverse selection is through information acquisition. In reality, a principal often adopts some form of interview or internship to learn signals about the agent's types, hence assisting her design of contracts with the agents. Through the algorithmic lens, we connect this optimal information acquisition problem in the principal-agent setup to the problem of function concavification, another fundamental concept in economics \cite{kamenica2011bayesian,aumann1995repeated} and mathematical optimization. This connection enables us to identify useful tractable cases for optimal information acquisition. 
Finally, we notice that the hardness of concavification can originate from the hardness of convex maximization, even when the optimization region is restricted to a simplex with a linear number of vertices. This insight probably leads to the hardness of optimal information acquisition to reduce adverse selection.

\subsection{Organizations \& Contributions} 
The core conceptual contribution of this paper lies in elucidating the \emph{connections} between various economic models. We start by presenting a generalized framework in Section \ref{sec:setup} that extends the general principal-agent game model by \citet{myerson1982optimal} and unifies several important principal-agent problems, including contract design, Stackelberg games, and information design. 
In Section \ref{sec:opt}, we show that the optimal coordination mechanism can nevertheless be efficiently computed in this general model.
In Theorem \ref{thm:revelation-principle}, we extend the revelation principle from~\cite{myerson1982optimal} to the \emph{succinct revelation principle}, which ensures the existence of an optimal coordination mechanism that is succinct, truthful, and direct. 
Specifically, we introduce the succinctness constraint as an additional refinement of the optimal coordination policy, which uses at most one principal strategy to induce a specific agent type to take a specific action. This result leads to our Theorem~\ref{thm:computability}: with the {succinct revelation principle}, the optimal coordination policy can be formulated as a convex program and solved efficiently in regular cases. We also identify the irregular cases where the exact optimal mechanism does not exist and provide an efficient algorithmic approach to determine the $\epsilon$-optimal solution.

Section \ref{sec:special-cases} carries the most conceptual insights in this paper. Zooming into several popular special cases of principal-agent games, we showcase how various (and seemingly irrelevant) economic designs  there can all be viewed as a coordination mechanism here, up to different levels of restrictions, e.g., whether the principal's strategy can coordinate with the agent's action or reported type. We systematically review the algorithmic properties of these special design problems, compare the nuanced complexity differences under their distinct problem structure, and settle the few remaining gaps in the literature with our Theorem~\ref{thm:computability} and~\ref{thm:Optimal_IC_leader_strategy}.

Last but not least, in Section~\ref{sec:info-acq}, we consider costly information acquisition problems to reduce adverse selection and hence assist the design of optimal coordination mechanisms. In Theorem~\ref{thm:computing-optimal-info-scheme}, leveraging techniques from function concavification, we characterize the general problem structure where information acquisition can be computed efficiently. This result notably implies the algorithmics of the costly persuasion problem~\cite{gentzkow2014costly} and has broad implications for the information acquisition problem in general. Meanwhile, we prove in  Theorem~\ref{thm:hard-info-acquisit} that optimal information acquisition in principal-agent problems can be APX-hard in general. This hardness result also suggests that even concavifying the difference between two convex functions is generally intractable despite its simplistic structure. The two results together shed light on the computational properties of function concavification through the algorithmic lens.

\subsection{Related Work} 

The principal-agent problem, with its intricate dynamics and multifaceted applications, spans a long and extensive research literature that spans multiple disciplines \cite{ross1973economic,holmstrom1979moral,mirrlees1976optimal,harris1979optimal,myerson1982optimal,grossman1992analysis, kadan2017existence}. Among them, \citet{myerson1982optimal} introduced a framework for the \textit{generalized} principal-agent problem, abstracting individuals' decisions to a higher level while losing some of the structural properties of specific principal-agent problems. \citet{kadan2017existence} considers the same framework and shows the condition where optimal mechanism exists. Next, we discuss a few works from several different domains that are closely related to our paper.

\vspace{2mm}
\noindent\textbf{Contract Design.} The contract theory has been a crucial branch of economics \cite{grossman1992analysis, smith2004contract, laffont2009theory}. 
Driven by an accelerating trend of contract-based markets deployed to Internet-based applications, the contract design problem recently started to receive a surging interest, especially from the computer science community \cite{dutting2019simple, dutting2021complexity, guruganesh2021contracts, alon2021contracts, castiglioni2021bayesian, castiglioni2022designing}.
As pointed out by \citet{alon2021contracts}, this includes online markets for crowdsourcing, sponsored content creation, affiliate marketing, freelancing and etc. The economic value of these markets is substantial and the role of data and computation is pivotal.  Hence, to accommodate the growing need for data elicitation and personalization, the aspect of private agent information in principal-agent problems has been studied by several recent works.  \citet{guruganesh2021contracts} modeled the hidden type of the agent which determines the outcome distribution and presents hardness results on solving the optimal menu of a single contract. \citet{alon2021contracts} modeled the hidden type of the agent which determines the agent's cost per unit-of-effort as a single-parameter setup and showed its tractability under a constant number of actions. \cite{castiglioni2022designing} presented a positive result on the efficient computation of the optimal menu of randomized contracts while showing that the exact optimal solution does not always exist.
Our work presents a simplified and unified characterization of the optimal coordination mechanism in the more generalized model of the principal-agent problem. 
Notably, our algorithm \ref{alg:find-neighbour} solves for the optimal (or $\epsilon$-optimal) mechanism in explicitly constructed convex programs, saving us from the (practically inefficient) ellipsoid method in \cite{castiglioni2022designing}.

\vspace{2mm}
\noindent\textbf{Bayesian Persuasion.} In a Bayesian persuasion problem \cite{kamenica2011bayesian}, the principal (sender) possesses private information about the state of the world and their goal is to influence the agent (receiver)'s actions by selectively revealing the information. \citet{castiglioni2022bayesian} extended the classic Bayesian persuasion problem with a type reporting step in which the receiver is asked to report their type to the sender, which corresponds to a specific instantiation of our generalized principal-agent framework into the Bayesian persuasion application. They showed both positive and negative results under different restrictions on the principal's design space, which we will discuss in more detail in Section \ref{sec:special-cases}. \citet{kolotilin2017persuasion} also studied persuasion mechanisms with type reporting, while they focused on the setting where the agent only has two actions.

\vspace{2mm}
\noindent\textbf{Stackelberg Games.} 
Stackelberg games \cite{stackelberg1934marktform,von2004leadership} have a wide range of applications in economics, public policy-making, and security \citep{von2004leadership,von2010leadership,van2010dynamic,paruchuri2008playing,tambe2011security,kiekintveld2013security}. \citet{conitzer2006computing} modeled the hidden type of the agent and showed the NP-Hardness to compute the optimal principal mechanism, termed Bayesian Stackelberg equilibrium (BSE),  where the principal is only allowed to play a single strategy for all agent types. \citet{paruchuri2008playing} then developed an efficient mixed integer linear programming tool for practical computation of the BSE. Our work considers a generalized framework, which enables personalized principal strategies for different agent types. Moreover, these principal strategies can be randomized. This concept of randomized commitment has also been investigated by \citet{gan2019imitative} to efficiently compute optimal mechanisms to play against imitative deception in Stackelberg games.

\vspace{2mm}
\noindent\textbf{Information Acquisition.} Information acquisition for reducing adverse selection in principal-agent problems has been widely studied in the literature, particularly in asset markets where  the buyer often  acquires costly  signals about assets' types from outside sources, 
\cite{jang2021adverse,maug1998large,chemla2014skin} (see more detailed discussions in Appendix \ref{appendix_sec:related_work}). These studies analyze the equilibrium structure and properties of the induced game, whereas our study is about algorithms. 
information acquisition  has also been studied in mechanism design to refine uncertainty \cite{bergemann2002information,barlevy2000information}. 
Incentive-compatible data acquisition has been widely studied by the EC community \cite{chen2018optimal,chen2019prior}, though data have different properties from the information. Information elicitation is different since it is about eliciting truthful information without optimization over the payment. An exception is the recent work on optimizing payment of scoring rules \cite{Liopt-score,papireddygari2022contracts,pmlr-v195-hartline23a}, which turns out to relate to contract design. \citet{chen2011information} studied how to elicit information for decision-making. \citet{Cacciamani23} took an online learning perspective. \citet{oesterheld2020minimum} considered regret minimization with unknown agent preferences. The focus of all these works is on the tension between the agents' misreporting of information or type, often due to involved payment or preference. In contrast, our problem is purely about the tradeoff between the benefit and cost of information and is free of incentive issues. Nevertheless, we show that such basic tension is already NP-hard to address in general.

\section{The Generalized Principal-Agent Problem}
\label{sec:setup}
\subsection{Problem Setting}\label{sec:setup:problem}
We consider a generalized principal-agent problem that extends the model of \citet{myerson1982optimal}. The model features a sequential game played between two players, namely the \emph{principal}   and the \emph{agent}. %
The agent has a private type $\theta$ that is sampled from a probability distribution $f$. 
The agent also takes a private action $a$, which the principal cannot control and hence has to incentivize. We denote the type and action spaces of the agent as $\Theta$ and $\cA$, respectively. Meanwhile, the principal  can also execute a strategy $\bx$ from the feasible strategy space $\cX$.  These parameters together determine the utility for both the principal and agent,   denoted as $U$ and $V$  respectively, where $U,V: \cX\times\cA\times \Theta \to \RR$. As a conventional assumption, we assume $f, U, V, \Theta, \cA, \cX$ are   public knowledge, whereas the agent's type $\theta$ and action $a$ are only privately known to the agent. 

The original model of  \citet{myerson1982optimal} assumes that   $\Theta, \cX, \cA$ are all finite, for which he developed a linear programming approach to solving for the optimal principal mechanism. In this paper, we generalize his model to the case where $\cX \subseteq \RR^d$ is an arbitrary \emph{closed convex set}  with a linear agent utility structure and concave principal utility. Notably, despite that $\Theta,   \cA$ remain   finite in our model, we show in Section \ref{sec:special-cases} that this generalization already includes several major economic design problems as special cases. 
Specifically, we assume $U(\bx, a;\theta)$ is concave in $\bx$ for any given (discrete) $a, \theta$, whereas the agent's utility function $V$ satisfies the following \emph{linearity}. 

\begin{assumption}[Affine Agent Utility]\label{ass:linear}
Agent's utility function $V(\bx, a;\theta)$ satisfies %
\begin{equation}\label{eq:linear-ass}
   \lambda  \cdot V(\bx, a;\theta) + (1-\lambda) \cdot V(\bx', a;\theta) = V(\lambda  \bx + (1-\lambda) \bx', a;\theta), \quad \forall \bx, \bx' \in \cX, \lambda \in [0, 1],  \forall \theta, a  
\end{equation}
\end{assumption}
Note that, Assumption \ref{ass:linear} is equivalent to that for any $a, \theta$, there exists $\alpha^{a, \theta} \in \RR^d$ and $\beta^{a, \theta} \in \RR$ such that $V(\bx, a;\theta) = \alpha^{a, \theta} \bx + \beta^{a, \theta} $. We avoided this more direct definition in order to eliminate additional notations of $\alpha^{a, \theta}, \beta^{a, \theta}$, as the property in Equation \eqref{eq:linear-ass} suffices for our technical development.

\vspace{2mm}
\noindent\textbf{The Coordination Mechanism.} To find the optimal mechanism for the principal --- coined the \emph{coordination mechanism} by \citet{myerson1982optimal}, as we will also adopt throughout the paper ---  a standard revelation principle style argument shows that it is without loss of generality to focus  {\em direct} coordination mechanism. Such a mechanism is given by a coordination policy $\pi: \cX\times\cA\times \Theta \to [0,1]$ such that $\sum_{\bx \in \cX, a\in \cA} \pi(\bx, a; \theta) = 1$ for all $\theta \in \Theta$.\footnote{To be more precise, $\pi$ should have been thought of as a density function since $\bx$ is continuous and the summation condition should have been $\sum_{a\in \cA} \int_{\bx \in \cX } \pi(\bx, a; \theta) = 1$. However, as we will show below, for the purpose of principal utility maximization, it is  without loss of generality to focus on $\pi$ that supports on only finitely many $\bx$. Due to this reason, we will always use $\sum_{\bx \in \cX }$, instead of $\int_{\bx \in \cX }$, throughout the paper for notational convenience. } Given such a $\pi$, if the agent reports any type $\theta$, then with probability $\pi(\bx, a; \theta)$ the principal plays strategy $\bx$ and recommends the agent to take action $a$.
In summary, the sequence of the interactions proceeds as follows, 
\begin{enumerate}
\item The principal announces a coordination mechanism $\pi$ to the agent.
\item An agent type $\theta$ is drawn from $f$ and the agent reports a type $\theta'$ to the principal.
\item A strategy $\bx$ of the principal and an action $a$ of the agent are drawn from $\pi(\cdot, \cdot; \theta')$. The principal plays $\bx$ (and recommends the agent to take action $a$).
\item The agent plays an action $a'$ in response to $\bx$; that is, $a' \in \argmax_{b \in \cA} V(\bx, a; \theta)$.
\item The principal and the agent receive utilities $U(\bx, a'; \theta)$ and $V(\bx, a'; \theta)$, respectively.
\end{enumerate}

A direct coordination mechanism is {\em incentive compatible (IC)}, or {\em truthful}, if and only if the honest-obedient participation strategy---i.e., reporting his type truthfully and meanwhile also taking the recommended action---is optimal for the agent \cite{myerson1982optimal}. 
Mathematically, a truthful and direct coordination mechanism satisfies the following constraints:
$$
 \sum_{\bx \in \cX} \sum_{a \in \cA} \pi(\bx, a; \theta) \cdot V(\bx, a; \theta) \geq \sum_{\bx \in \cX} \sum_{a \in \cA} \big[ \pi(\bx, a; \theta') \cdot \max_{a'\in \cA} V(\bx, a'; \theta) \big], \quad \forall \theta, \theta'\in \Theta.
$$
Consequently, given a truthful and direct coordination mechanism, the principal's utility can be written as
$$ \sum_{\theta \in \Theta} f(\theta) \sum_{\bx \in \cX} \sum_{a \in \cA} \pi(\bx, a; \theta) \cdot U(\bx, a; \theta). $$

\vspace{2mm}
\noindent\textbf{Supplemental Constraints.} %
Another novel aspect of our framework is that it can accommodate supplemental convex constraints   on the coordination policy $\pi$, in the following format
\begin{equation}
\label{eq:supplimental-const}
\sum_{x\in \cX}\sum_{a\in \cA} \pi(\bx, a;\theta) \cdot \bx \in \cC_{\theta}, \qquad \forall \theta \in \Theta  
\end{equation}
for some convex set $\cC_{\theta}$. To see how this could be useful, in the Bayesian persuasion problem, $\bx$ will correspond to the agent's posterior belief over certain random state of nature and the well-known Bayes plausibility \cite{aumann1995repeated} constraint will have the format of  $\sum_{x\in \cX}\sum_{a\in \cA} \pi(\bx, a;\theta) \cdot \bx = \mu_{\theta}$ where $\mu_{\theta}$ is agent type $\theta$'s prior belief over the state (different types could have different beliefs). Such constraints can be captured by the above supplemental constraints. 
In Section \ref{sec:opt}, we will develop a revelation principle for the general principal-agent problem under supplemental constraints of the form \eqref{eq:supplimental-const}, and then show how to solve it efficiently through a carefully tailored algorithm.

\subsection{Warm-up: Myerson's Finite Action Case} 
Before presenting our results, it is worth recalling Myerson's solution for the case with a finite principal strategy set $\cX$ and without any supplemental constraint.
His solutions are based on the  following two results. 

\begin{prop}[Revelation Principle~\cite{myerson1982optimal}] \label{prop:revelation-principle}
There always exists an optimal coordination mechanism that is truthful and direct.
\end{prop}

\begin{prop}[\cite{myerson1982optimal}]
\label{prop:computability}
When the agent type set $\Theta$ and the players' action sets $\cX$ and $\cA$ are finite, the optimal truthful and direct coordination mechanism can be computed by the following linear program:
\begin{lp*}
    \maxi{\sum_{\theta \in \Theta} f(\theta) \sum_{x \in \cX} \sum_{a \in \cA} \pi(\bx, a; \theta) \cdot U(\bx, a; \theta)  }
\st  
\qcon{\sum_{\bx \in \cX} \sum_{a \in \cA} \pi(\bx, a; \theta) \cdot V(\bx, a; \theta) \geq \sum_{\bx \in \cX} \sum_{a \in \cA} \big[ \pi(\bx, a; \theta') \cdot \max_{a'\in \cA} V(\bx, a'; \theta ) \big]}{\theta, \theta'\in \Theta}
    \qcon{ \sum_{\bx\in \cX} \sum_{a\in \cA}   \pi(\bx, a; \theta) = 1 }{\theta \in \Theta} 
\end{lp*}
\end{prop}

Note that the revelation principle presented in Proposition~\ref{prop:revelation-principle} is  applicable even in  settings with infinite $\Theta$, $\cX$, and $\cA$, but Proposition~\ref{prop:computability} relies on finite action and type spaces.

\section{The Optimal Coordination Mechanism}\label{sec:opt}

In this section, we extend Myerson's result above and show that the optimal coordination mechanism for the principle can be computed even when $\cX$ is an arbitrary convex set and when supplemental constraints are imposed, as in the model of Section \ref{sec:setup:problem}. To do so, we first prove a  strengthened   revelation principle for our setup by leveraging convexity, and prove that it is without loss of generality to design the \emph{succinct coordination mechanism} for our setup. 

\begin{definition}\label{def:succint}
A direct coordination mechanism is \emph{succinct} if it satisfies the following  constraint:
$$ \big|\supp (\pi(\cdot, a; \theta) )\big| \leq 1, \quad \forall \theta \in \Theta, a\in \cA,$$ 
where $\supp(\cdot)$ denotes the support of a distributional function $\pi(\cdot, a; \theta) $. 
\end{definition}

In other words, a succinct mechanism prescribes deterministically a single strategy $\bx$ for each type $\theta$ reported and action $a$ chosen by the agent. 
Under this definition, a succinct mechanism admits the following simplified representation.
Denote by $\bx^{a, \theta}$   the strategy prescribed for a reported type $\theta$ and action $a$; namely,   $\bx^{a, \theta}$ is the strategy with $\pi(\bx, a; \theta) >0$   in the coordination mechanism $\pi$ and Definition \ref{def:succint} guarantees that there is exactly  one such $\bx$ for any fixed $a, \theta$ with $\pi(\bx, a; \theta) >0$.  

A succinct mechanism can thus be represented by $ \{( \pi(a; \theta), \bx^{a, \theta}) : \theta \in \Theta, a \in \cA\}$, subject to $\sum_{a\in \cA} \pi(a; \theta) = 1$   for all $\theta \in \Theta$. This mechanism is then implemented as follows: when the agent reports a type $\theta$, with probability $\pi(a; \theta)$, the principal plays $\bx^{a,\theta}$ and recommends the agent to take action $a$. The crucial advantage of succinct mechanisms is that the  entire mechanism $ \{( \pi(a; \theta), \bx^{a, \theta}) : \theta \in \Theta, a \in \cA\}$ now has a compact polynomial-size representation, irrespective of the size of $\cX$. However, it is still unclear whether it is without loss of generality to only consider such succinct direct coordination mechanisms for the principal's utility maximization. Our first main result shows that the answer is Yes --- and this is also the reason that we have been using the notation of $\sum_{\bx}$ instead of $\int_{\bx}$. We remark that Theorem \ref{thm:revelation-principle}  is a \emph{strengthened} version of Proposition \ref{prop:revelation-principle}, while not  a generalization. In fact, Theorem \ref{thm:revelation-principle} does not hold in the case of discrete $\cX$ and is a consequence of the convexity of player utilities in the feasible principal strategy set. We refer readers to Appendix \ref{appendix_sec:thm_rp} for the complete proof.

\begin{theorem}[Succinct Revelation Principle] \label{thm:revelation-principle}
In any generalized principal-agent problem with supplemental constraints, there always exists an optimal coordination mechanism that is succinct, truthful, and direct.
\end{theorem}

Thanks to Theorem \ref{thm:revelation-principle}, we can restrict our attention to finding the optimal succinct mechanism, represented by  variables $ \{( \pi(a; \theta), \bx^{a, \theta}) : \theta \in \Theta, a \in \cA\}$. This representation leads to the following (non-convex) optimization problem (OP):
\begin{lp} \label{lp:original}
    \maxi{\sum_{\theta \in \Theta} f(\theta) \sum_{a \in \cA} \pi(a; \theta) \cdot U(\bx^{a, \theta}, a; \theta)   }
\st 
\qcon{\sum_{a \in \cA} \pi(a; \theta) \cdot V(\bx^{a, \theta}, a; \theta) \geq \sum_{a \in \cA} \big[ \pi(a; \theta') \cdot \max_{a'\in \cA} V(\bx^{a, \theta'}, a'; \theta ) \big]}{\theta, \theta'\in \Theta}
    \qcon{\bx^{a,\theta} \in \cX;\quad \pi(a; \theta) \in [0,1]}{a\in \cA, \theta \in \Theta}
    \qcon{\sum_{a\in \cA} \pi(a; \theta) = 1; \quad \sum_{a\in \cA} \pi(a;\theta) \cdot \bx^{a,\theta} \in \cC_{\theta}}{\theta \in \Theta} 
\end{lp}

In the above program, the objective is the principal's expected utility with respect to the prior distribution $f$ over the agent's private type when deploying the succinct mechanism $\{( \pi(a; \theta), \bx^{a, \theta}) : \theta \in \Theta, a \in \cA\}$. The first set of constraints is the key incentive constraint, which guarantees that under the succinct mechanism, the agent is incentivized to report their type $\theta$ truthfully and respond with action $a$ when the principal plays strategy $x^{a, \theta}$. That is, the agent's utility (on the RHS) under truthfulness is no worse than the agent's utility (on the LHS) under any alternative report of type $\theta'$ and the action $a'$ that maximizes the agent's utility under the reported type $\theta'$. 
The remaining constraints ensure the feasibility of the succinct mechanism. The second and third set of constraints ensures that the principal's strategy is within the feasible strategy space $\cX$ and the randomization of the principal's coordination policy is valid. The last set of constraints ensures that any supplemental constraint is satisfied. 

Unfortunately, OP \eqref{lp:original} is a non-convex quadratic program due to the bilinear term $\pi(a;\theta) \cdot \bx^{a,\theta}$ (recall that $V$ is  an affine function of $\bx$ by Assumption \ref{ass:linear}). Here then comes our second main result, which surprisingly shows that the above bilinear program can be solved efficiently by a carefully designed algorithm.
\begin{theorem}\label{thm:computability}
In the generalized principal-agent game, the optimal coordination mechanism can be computed in $\poly(|\cA|, |\Theta|,d)$ time, where $d$ is the dimension of  $\cX (\subseteq \RR^d)$.  %
\end{theorem}

The algorithm for Theorem \ref{thm:computability}  simply  solves carefully constructed polynomial-size convex programs, which is  also  practical in reality.  However, the correctness  proof of the algorithm turns out to be quite involved. Such non-triviality is perhaps as expected --- \citet{castiglioni2022designing} studied the same computational problem but for the contract design problem which is a  special case of our model; yet the algorithm they designed for that special case is  even   more complex  and had to go through the machinery of the ellipsoid method and counting polytope vertices, leading to a polynomial time algorithm with large polynomial degrees. Our proof uses a completely different approach, and is simpler compared to that of  \cite{castiglioni2022designing}, yet is applicable to much more general problem setups. Specifically, we show that for  typical  problem instances --- i.e., when $\cX$ is a bounded set, or when the instance is ``regular'' (see formal definition in our proof) --- OP \eqref{lp:original}  reduces to an explicitly constructed convex program (CP) and thus admits practically efficient algorithms. The key challenge in proving Theorem \ref{thm:computability} is to address unbounded $\cX$ with irregular game structure (the special case of contract design as studied in \cite{castiglioni2022designing} happens to fall within this difficult situation).  To handle this case, we develop approaches to first identify the irregularity and then adjust the mechanism to output a finite solution that is arbitrarily close to the principal's optimal utility.\footnote{An arbitrarily small utility loss turns out to be necessary, as \cite{castiglioni2022designing} also shows that the precisely optimal principal utility cannot be achieved in contract design unless $\infty$ payment is allowed. }  The formal proof is presented below in subsection \ref{sec:alg:proof}.%

\subsection{Proof Sketch of Theorem \ref{thm:computability}}\label{sec:alg:proof}

At a high level, we will convert OP~\eqref{lp:original} into a convex optimization problem, which is efficiently solvable.  
For this purpose, define the variable $ \bz^{a, \theta} \coloneqq \pi(a; \theta) \cdot \bx^{a, \theta} $ and the following joint space of $\pi(a; \theta)$ and $\bz^{a,\theta}$:    
\[
\cP \coloneqq \big\{ ( \lambda, \lambda\bx) : \bx \in \cX, \lambda\in [0,1] \big\}.
\]
\noindent
As shown in Lemma~\ref{lm:convex-z} below, a useful property of $\cP$ is that it is convex whenever $\cX$ is convex. The proof of Lemma 1 can be found in Appendix \ref{appendix_sec:thm_computability}.

\begin{lemma}\label{lm:convex-z}
If $\cX$ is a convex set, then $\cP := \{ ( \lambda, \lambda\bx) : \bx \in \cX, \lambda\in [0,1] \}$ is a convex set.
\end{lemma}

Nevertheless, $\cP$ may not be a closed set even when $\cX$ is closed. For instance, if $\cX = \RR_{\geq 0}$, then $\cP = \{(\lambda, z): \lambda \in (0, 1], z\geq0\} \cup \{(0,0)\}$ which is convex but not closed. This may cause irregularity issues in our optimization and a large part of our proof is devoted to addressing these issues. 

Next, we convert OP~\eqref{lp:original} to a convex optimization problem, possibly one with a feasible region that is not closed. We define $\hat{U}, \hat{V}$ such that $U,V$ are perspective functions of $\hat{U}, \hat{V}$ respectively, i.e., 
$$ \hat{U}(t \bx, t, a; \theta) \coloneqq t \cdot U(\bx, a; \theta),\quad  \hat{V}(t  \bx, t, a; \theta) \coloneqq t \cdot V(\bx, a; \theta). $$ 
Recall that perspective transformation preserves the convexity~\cite{boyd2004convex}. 
Since $U$ is concave in $\cX$, the new function $\hat{U}$ is also concave w.r.t. $t \cdot \bx$ and $t$. Similarly, $\hat{V}$ is linear w.r.t. $t \cdot \bx$ and $t$. 
We can then obtain the following convex optimization with variables 
$ \{( \pi(a; \theta), \bz^{a, \theta}) : \theta \in \Theta, a \in \cA\} $. 
For ease of description, we will also refer to the solution of this optimization problem as a mechanism $M$. 
\begin{lp}\label{lp:transformed}
    \maxi{\sum_{\theta \in \Theta} f(\theta) \sum_{a \in \cA}  \hat{U}(\bz^{a, \theta}, \pi(a; \theta), a; \theta)  }
\st 
\qcon{\sum_{a \in \cA}  \hat{V}(\bz^{a, \theta}, \pi(a; \theta), a; \theta) \geq \sum_{a \in \cA} \max_{a'\in \cA} \hat{V}(\bz^{a, \theta'}, \pi(a; \theta'), a'; \theta) }{\theta, \theta'\in \Theta}
    \qcon{(\pi(a; \theta), \bz^{a, \theta}) \in \cP }{a \in \cA, \theta \in \Theta} 
    \qcon{\sum_{a\in \cA} \pi(a; \theta) = 1; \quad \sum_{a\in \cA} \bz^{a,\theta} \in \cC_{\theta} }{\theta \in \Theta}
\end{lp}
\noindent
In CP \eqref{lp:transformed}, the first set of constraints is convex with respect to the variables $ \{( \pi(a; \theta), \bz^{a, \theta}) : \theta \in \Theta, a \in \cA\} $. 
The third set of constraints is linear. In Lemma \ref{lm:convex-z}, we show that the second set of constraint is convex. The last set of constraints is convex as $\cC_{\theta}$ is convex. In Lemma \ref{lm:lp-equivalence}, we show this program admits the same supremum as OP \eqref{lp:original} and their solutions can be converted into one another by construction. The proof of Lemma \ref{lm:lp-equivalence} is in Appendix \ref{appendix_sec:thm_computability}.

\begin{lemma} \label{lm:lp-equivalence}
$\sup (\textup{CP} \eqref{lp:transformed}) = \sup (\textup{OP} \eqref{lp:original})$.  That is, any mechanism $ \{( \pi(a; \theta), \bz^{a, \theta}) : \theta \in \Theta, a \in \cA\}$ in CP \eqref{lp:transformed} achieves the same amount of utility as some mechanism $ \{( \pi(a; \theta), \bx^{a, \theta}) : \theta \in \Theta, a \in \cA\}$ in OP \eqref{lp:original}, where $\bx^{a, \theta} = \begin{cases}\bz^{a, \theta} / \pi(a; \theta) & \pi(a; \theta) \neq 0 \\ 0 & \pi(a; \theta) = 0 \end{cases}$, for each $(a, \theta)$. %
\end{lemma}

By Lemma \ref{lm:lp-equivalence}, the optimal mechanism can be recovered from the optimal solution to the convex program (CP) \eqref{lp:transformed}, except for a technical caveat --- the feasible region CP \eqref{lp:transformed} may not be a closed set. Thus when the optimal solution of CP \eqref{lp:transformed} lies on the boundary not contained in the feasible region, the exact optimal coordination mechanism may not exist.
This makes it difficult to apply an off-the-shelf convex program solver. 
We address this issue through a careful analysis below.  

\vspace{2mm}
\noindent \textbf{The Case with Bounded Action Set.}
For simplicity, we start by assuming that the set $\cX$ is bounded. 
In this case, we show via Lemmas \ref{lm:convex-z} and the following Lemma \ref{lm:compact-z} (proved in Appendix \ref{appendix_sec:thm_computability}) that the set $\cP$ in the second set of constraints is  convex and closed, and therefore the solution to CP \eqref{lp:transformed} is the optimal coordination mechanism. Moreover, we can convert the optimal solution from CP \eqref{lp:transformed} to the original mechanism according to Lemma \ref{lm:lp-equivalence}.

\begin{lemma}\label{lm:compact-z}
If $\cX$ is a compact set, then $\cP := \{ ( \lambda, \lambda\bx) : \bx \in \cX, \lambda \in [0,1] \}$ is a compact set.
\end{lemma}

\vspace{2mm}
\noindent \textbf{The Case with Unbounded Action Set under Regularity Assumption.}
We now consider the case when $\cX$ is unbounded (but closed). 
The set $\cP$ is not a closed set and we cannot directly solve CP \eqref{lp:transformed}. 
We say that the problem instance is \emph{regular} if CP \eqref{lp:transformed} admits a maximum,
and it is \emph{irregular} otherwise.

For a regular game instance, instead of solving CP \eqref{lp:transformed}, we can without loss of generality solve its ``closure'' CP \eqref{lp:transformed-closure}, which is simply CP \eqref{lp:transformed} with its  second constraint relaxed to the closure $\bar{\cP}$ of the original (possibly not closed) $\cP$. Note that $\bar{\cP}$ is convex by the fact that the closure of a convex set must be convex. 
According to the definition of regularity above, the solution to CP~\eqref{lp:transformed-closure} is exactly the maximum of CP \eqref{lp:transformed}. Hence, we can again convert this solution to the original mechanism according to Lemma \ref{lm:lp-equivalence}.
\begin{lp}\label{lp:transformed-closure}
    \maxi{\sum_{\theta \in \Theta} f(\theta) \sum_{a \in \cA}  \hat{U}(\bz^{a, \theta}, \pi(a; \theta), a; \theta)  }
\st 
\qcon{\sum_{a \in \cA}  \hat{V}(\bz^{a, \theta}, \pi(a; \theta), a; \theta) \geq \sum_{a \in \cA}   \max_{a'\in \cA} \hat{V}(\bz^{a, \theta'}, \pi(a; \theta'), a'; \theta ) }{\theta, \theta'\in \Theta}
    \qcon{(\pi(a; \theta), \bz^{a, \theta}) \in \bar{\cP} }{a \in \cA, \theta \in \Theta} 
    \qcon{\sum_{a\in \cA} \pi(a; \theta) = 1; \quad \sum_{a\in \cA} \bz^{a,\theta} \in \cC_{\theta} }{\theta \in \Theta}
\end{lp}

\noindent\textbf{Addressing  Irregularity.} 
However, not all problem instances are regular. 
A concrete example of irregularity can be found   in \cite{castiglioni2022designing}, where the optimal menu of randomized contracts has to use an infinitely large payment. This is precisely an irregular problem instance under our definition. 
We denote $\cM$ as the feasible region of CP~\eqref{lp:transformed},  $\bar{\cM}$ as the feasible region of CP~\eqref{lp:transformed-closure} and recall our definition above: if the solution to CP \eqref{lp:transformed-closure} $M \not\in \cM$, the problem instance is irregular.  For the sake of convenience, we will also say that a pair $(a, \theta)$ is \emph{irregular} in a mechanism if $\pi(a; \theta)=0$ but $\bz^{a, \theta}\neq 0$; and it is \emph{regular} otherwise. 
In Lemma \ref{lm:unbounded-z}, we provide a characterization of the unattainable boundary of $\cP$, with which we can identify irregular cases and efficiently resolve them through an algorithmic approach. 
The proof of Lemma \ref{lm:unbounded-z} can be found in Appendix \ref{appendix_sec:thm_computability}. 

\begin{lemma}\label{lm:unbounded-z}
Suppose that $M$ is a feasible mechanism in CP \eqref{lp:transformed-closure}. Then $M$ is feasible in CP \eqref{lp:transformed} if and only if $(a,\theta)$ is regular in $M$ for all $a\in \cA, \theta\in \Theta$.

\end{lemma}

To solve irregular instances, we now show that, for any arbitrarily small $\epsilon > 0$, we can obtain an (additively) $ \epsilon$-optimal mechanism in polynomial time using Algorithm~\ref{alg:find-neighbour}, when the problem's optimal objective value is bounded (for situations with unbounded optimal utility, our constructed mechanism will also have unbounded utility).\footnote{Typical problems, such as contract design, are guaranteed to have bounded optimal utility even when the principal strategy is unbounded (due to the individual rationality constraint). } 
The high level idea is to solve CP \eqref{lp:transformed-closure} first and then add a slight perturbation to the optimal solution if it is on the unattainable boundary $\bar{\cP} \setminus \cP$. This is based on two established facts: (1) the concavity of the objective in CP \eqref{lp:transformed-closure}; (2) the convexity of constraints in CP \eqref{lp:transformed-closure}. Consequently,  we can construct a new mechanism $\bar{M} = (1-\epsilon) M^* +  \epsilon  M'$ for any arbitrarily small $\epsilon>0$, which is a convex combination of the optimal mechanism $M^*$ to CP~\eqref{lp:transformed-closure} and a coordination mechanism $M'$ that lies within $\cP$ of CP~\eqref{lp:transformed}.
$\bar{M}$ attains utility at least $(1-\epsilon)$ fraction of $M^*$ and lies within $\cP$ of CP~\eqref{lp:transformed} by convexity. Since $\bar{M}$ is feasible in CP~\eqref{lp:transformed}, we can convert it to {a solution to OP~\eqref{lp:original}} according to Lemma~\ref{lm:lp-equivalence}.

Hence, the remaining challenge is to find the mechanism $M'$, which lies within $\cP$.
The task is akin to finding a relative interior point in the feasible region of CP \eqref{lp:transformed-closure} 
(but not exactly the same as a feasible mechanism in CP \eqref{lp:transformed-closure} might be on the boundary $\pi(a;\theta)=0$ and $\bz^{a,\theta}=0$, e.g., when the action $a$ of type $\theta$ is strictly dominated by other actions). 
Our approach relies on the observation of Lemma \ref{lm:unbounded-z}
and proceeds by enforcing a margin on the constraints $\pi(a,\theta) \ge 0$ to nudge solutions on the boundary $\bar{\cP} \setminus \cP$ towards $\cP$.
Specifically, we will use CP~\eqref{lp:transformed-interior} to find out the maximum enforceable margin for each action-type pair $\pi(a;\theta)$.
More efficiently, leveraging the special problem structure, we can avoid having to always solve CP~\eqref{lp:transformed-interior} for all the action-type pairs; the approach is presented in Algorithm~\ref{alg:find-neighbour}.

\begin{lp}\label{lp:transformed-interior}
    \maxi{\pi(\tilde{a}; \tilde{\theta})  }
\st 
\con{\text{the same constraints as CP~\eqref{lp:transformed-closure}}}
\end{lp}

\begin{algorithm}[t]
\caption{Computing an $\epsilon$-optimal solution for CP \eqref{lp:transformed} }
\label{alg:find-neighbour}
\begin{algorithmic}[1] %
\STATE Solve for the optimal mechanism $M^*$ in CP \eqref{lp:transformed-closure};
\STATE $\cS \gets \emptyset$ and  $\bar{M} \gets  {M}^*$;
\WHILE[i.e., $ \pi(\tilde{a};\tilde{\theta})=0$ and $\bz^{\tilde{a},\tilde{\theta}} \neq \zero$]{exists $(\tilde{a}, \tilde{\theta})$ irregular in $\bar{M}$ }

\STATE $M^{\tilde{a}, \tilde{\theta}} \gets$ an optimal mechanism in CP \eqref{lp:transformed-interior}; 
\label{ln:M-gets}

\STATE $\cS \gets \cS \cup \{(\tilde{a},\tilde{\theta})\}$;

\STATE $\bar{M} \gets  (1-\epsilon) {M}^* + \frac{\epsilon}{|\cS|} \sum_{(a, \theta)\in \cS} M^{a, \theta}$;
\label{ln:bar-M-gets}
\ENDWHILE
\RETURN $\bar{M}$. 
\end{algorithmic}
\end{algorithm}

\noindent\textbf{Correctness and Time Complexity of Algorithm \ref{alg:find-neighbour}. }
Clearly, if Algorithm \ref{alg:find-neighbour} terminates, the solution must be a feasible $(1-\epsilon)$-approximate solution, as desired; indeed, the termination condition, that the output mechanism $\bar{M}$ contains no irregular pair $(\theta, a)$, implies that $\bar{M}$ is feasible in CP~\eqref{lp:transformed} by Lemma \ref{lm:unbounded-z}. 
The fact that the iterative process of Algorithm~\ref{alg:find-neighbour} always terminates can be verified by the fact that each action-type pair appears at most once in the iterations, which we demonstrate next.
This observation also implies that Algorithm~\ref{alg:find-neighbour} terminates in at most $|\cA||\Theta|$ many iterations.

Suppose that $(\tilde{a},\tilde{\theta})$ appears in some iteration, and consider the following cases with respect to the optimal mechanism  $M^{\tilde{a},\tilde{\theta}}$ at Line~\ref{ln:M-gets}:
\begin{itemize}
\item[1.] $\pi(\tilde{a};\tilde{\theta}) > 0$ in $M^{\tilde{a},\tilde{\theta}}$.
In this case, we will have $\pi(\tilde{a};\tilde{\theta}) > 0$ in  $\bar{M}$ in the remainder of the process as $M^{\tilde{a},\tilde{\theta}}$ always gets a positive weight at Line~\ref{ln:bar-M-gets}. Hence, $(\tilde{a},\tilde{\theta})$ remains regular and will not appear again in the subsequent iterations.

\item[2.] $\pi(\tilde{a};\tilde{\theta}) = 0$ in $M^{\tilde{a},\tilde{\theta}}$.
In this case, we argue that it must be $\bz^{\tilde{a},\tilde{\theta}} = \zero$, so $(\tilde{a},\tilde{\theta})$ should actually not be selected in the previous iterations, meaning that this case is not possible.
To see that $\bz^{\tilde{a},\tilde{\theta}} = \zero$, suppose for the sake of contradiction that $\bz^{\tilde{a},\tilde{\theta}} \neq \zero$.
Since $M^{\tilde{a},\tilde{\theta}}$ is a feasible solution to CP~\eqref{lp:transformed-interior}, we have $p = \left(\pi(\tilde{a};\tilde{\theta}), \bz^{\tilde{a},\tilde{\theta}} \right) \in \bar{\cP}$, so according to the definition of closure there exists a point $p' = (\lambda, \bz) \in \cP$ in every neighborhood of $p$.
Now that $\bz^{\tilde{a},\tilde{\theta}} \neq \zero$, there must be some neighbourhood of $p$ that does not contain $(0, \zero)$, which then implies the existence of $p' = (\lambda, \bz) \in \cP \setminus \{(0, \zero)\}$. 
According to the definition of $\cP$, we have $\lambda > 0$, which contradicts the assumption that $\pi(\tilde{a};\tilde{\theta}) = 0$ in the optimal solution to CP~\eqref{lp:transformed-interior}.
\end{itemize}

\noindent\textbf{Final Remarks.}
While the proof of Theorem \ref{thm:computability} is involved, the main challenge in the analysis lies in handling the irregular situations in which the exact optimal solution may not be attainable. Though this is important for theoretical guarantee, we expect such corner situations not to happen often. For typical regular instances, a practical implementation of our algorithm could be very efficient as it suffices to solve CP~\eqref{lp:transformed-closure}. 
If there is no irregular $(a, \theta)$ pair in the solution, such that $\pi(a; \theta) = 0$ and $\bz^{a,\theta}\neq 0$, then one can directly convert the solution into a feasible coordination mechanism according to Lemma \eqref{lm:lp-equivalence}. Otherwise, to find the $\epsilon$-optimal mechanism, our algorithm would just solve a few more convex programs unless the instance is extremely degenerated.

\section{Notable Special Cases of  Generalized Principal-Agency}
\label{sec:special-cases}
This section  discusses several notable applications of our generalized principal-agent model. Our goal here is to establish connections among many familiar (and important) models in the literature. This also offers us an opportunity to overview and compare many previous algorithmic results, which were developed separately by different research threads that may appear irrelevant at first glance but could now be examined under the same umbrella. In order to paint a complete picture of the algorithmic landscape of our problem and its special cases,  we also establish new computational results along the way, which completes the last few missing pieces of the whole picture.

\vspace{2mm}
\noindent \textbf{Coordination Mechanisms in Restricted Design Spaces. } Our Theorem  \ref{thm:computability} develops an efficient algorithm for solving the globally optimal coordination mechanism in the general principal-agent problem. This readily implies a unified polynomial time algorithm  for finding the globally optimal mechanisms  in  all the special cases. Unsurprisingly, the optimal mechanisms for each special setting have  been studied in previous literature, though under different terms (e.g., contracts, signaling schemes, Stackelberg equilibrium, etc.). What is interesting though is that previous studies have paid more attention to the design in \emph{restricted} design space, due to either practical motivations or historical reason as some of these restricted mechanisms are more well-motivated in reality. 

Specifically, in previous algorithmic literature, two natural restrictions to generic  coordination mechanisms $\pi(\bx, a; \theta)$ have been widely studied:
\begin{itemize}
    \item \textbf{Action-independent Coordination Mechanisms}, with succinct representation $ \{ \bx^{\theta} \}_{\theta \in \Theta}$.  
    
    These mechanisms are allowed to use different strategies for different agent types, but the agent will not receive an action recommendation in this coordination mechanism --- i.e., the generic strategy   degenerates from $\pi(\bx, a; \theta)$ to $\pi(\bx; \theta)$. Due to the linearity of the agent's utility function and the restriction   that the agent does not receive any action recommendation, this means that the agent will induce the same utility under ``aggregated'' strategy $\bx^{\theta} = \sum_{\bx \in \cX} \pi(\bx; \theta) \cdot \bx $; consequently, the concavity of the principal's utility function implies that the resultant succinct mechanism of form  $ \{ \bx^{\theta} \}_{\theta \in \Theta}$ is no worse hence can always achieve the optimal principal utility. Note that this observation holds even when the supplemental constraints of form $\sum_{x\in \cX}\sum_{a\in \cA} \pi(\bx, a;\theta) \cdot \bx \in \cC_{\theta}$ present. 
    
    \item \textbf{Type-independent Coordination Mechanisms}, with succinct representation $\bx$ without supplemental constraints and  $ \{ \pi(\bx) \}_{\bx \in \cX}$ with  supplemental constraints. 
    
    These mechanisms  have to use the same recommendation for every agent type hence the  generic (possibly not IC) coordination mechanism    degenerates from $\pi(\bx, a; \theta)$ to $\pi(\bx, a)$.\footnote{Note that $\pi(\bx, a)$ is not incentive compatible in general since the recommendation of action $a$ is likely not obedient for every receiver type $\theta$. However, it is well-defined as a generic coordination mechanism.} In this case, no type reporting is needed since the strategy is not discriminative. Moreover, in such a mechanism, the distribution of $\bx$  conditioned on the receiver's action recommendation $a$, i.e., $\Pr(\bx|a)$ induced by $\pi(\bx,a)$, is unrelated to   $\Pr(\bx|a')$, hence the principal can simply maximize each distribution $\Pr(\bx|a)$  separately and then pick the most profitable one among them. This reduces the problem to finding a single distribution $ \{ \pi(\bx) \}_{\bx \in \cX}$. 
    
    The above mechanism $ \{ \pi(\bx) \}_{\bx \in \cX}$ can be further simplified when there is no supplemental constraint in the form of Equation  \eqref{eq:supplimental-const}. In this case, the principal can simply examine her utility under every $\bx$ in the support of $\pi(\cdot)$ and simply adopt the single strategy $\bx^*$ that achieves her best utility. However, this reduction is not valid when there are supplemental constraints because the single $\bx^*$ may no longer be feasible under supplemental constraints.   
    These differences  will be reflected in our different special cases below. 
\end{itemize}

Next, we will elaborate on a few notable special cases of our general model. For each case, we briefly survey known results or, in some cases, prove new computational results for finding restricted or unrestricted coordination mechanisms. We defer detailed descriptions of every specific model (e.g., agent's utility function, action space, etc.) to Appendix \ref{appendix:additional_special}.

\subsection{Bayesian Contract Design: A Case with Unbounded Strategy Space}
One special case of our generalized principal-agent problem is  the Bayesian contract design problem, as populated in recent works \cite{alon2021contracts,guruganesh2021contracts}. 
The three classes of coordination mechanisms above have all been studied in the above contract design problem, known under the following terminologies:
\begin{itemize}
    \item   \emph{Menu of randomized contract} \cite{castiglioni2022designing} which corresponds to our general succinct coordination mechanism;
    \item \emph{Menu of deterministic contracts} \cite{guruganesh2021contracts} which corresponds to action-independent coordination mechanism in succinct form, i.e., $\{ \bx^{\theta} \} $, which is   without loss of generality as mentioned above;  
    \item \emph{Single contract} \cite{xiao2020optimal,alon2021contracts} which corresponds to  type-independent coordination mechanism in succinct form, i.e., a single contract $\bx^*$, which is without loss of generality.    
\end{itemize} The following Table \ref{tab:hardness-results1} summarizes the algorithmics of each mechanism.    

\begin{table}[tbh]
    \footnotesize
    \centering
    \begin{tabular}{|c|c|c|c|} \hline
         & \multirow{2}{*}{\makecell{Single Contract}} & Menu of    & Menu of  \\
        &  &  Deterministic Contracts   &  Randomized Contracts
         \\ \hline
      \multirow{2}{*}{\makecell{When   type only affects \\ action cost }}   &  APX-Hard  &  APX-Hard  & Polynomial \\
          &  \cite{xiao2020optimal} &  \makecell{\cite{xiao2020optimal} \\ \cite{castiglioni2022designing}} & \makecell{\cite{castiglioni2022designing} \\ $[$Our Thm. \ref{thm:computability} $]$}  \\ \hline
      \multirow{2}{*}{\makecell{When   type only affects \\ outcome distribution}}   &  APX-Hard &  APX-Hard  & Polynomial \\ 
       &  \cite{guruganesh2021contracts} &  \makecell{\cite{guruganesh2021contracts} \\ \cite{castiglioni2022designing}} &  \makecell{\cite{castiglioni2022designing} \\ $[$Our Thm. \ref{thm:computability}  $]$}\\
      \hline
    \end{tabular}
    \caption{Complexity landscape for the case of optimal Bayesian contract design. %
    } 
    \label{tab:hardness-results1}
\end{table}
\vspace{-6mm}
\begin{remark}
We remark that the optimal positive results in Table \ref{tab:hardness-results1} is the case of menu of randomized contracts. As mentioned previously, this positive result has been proved in \cite{castiglioni2022designing} via a  different and more complex approach relying on the ellipsoid method and polytope analysis. Our algorithm in Theorem \ref{thm:computability} is simpler, more practical yet more generally applicable. 
\end{remark}

\subsection{Bayesian Persuasion: A Case with supplemental Constraints}
\label{sec:info-design}
A second special case of our generalized principal-agent problem is  the (Bayesian) persuasion  of  a receiver with private  type $\theta$ drawn from discrete distribution $f$ supported on set $\Theta$, as studied in  \cite{kolotilin2017persuasion,castiglioni2022bayesian}. 
It turns out that two out of the three coordination mechanisms above are studied in the literature. 
\begin{enumerate}
    \item \emph{Private signaling} \cite{kolotilin2017persuasion,castiglioni2022bayesian} which corresponds to our general succinct coordination mechanism;
    \item \emph{Public signaling} \cite{kolotilin2017persuasion,castiglioni2022bayesian} which requires the sender to use the same signaling scheme for every receiver and corresponds to our type-independent coordination mechanism.
\end{enumerate}
Why the action-independent coordination  was not studied? It turns out that, in persuasion problems, all action-independent coordination mechanisms  are equivalent to  the trivial signaling scheme of revealing no information. This is reflected in our descriptions of action-independent coordination at the beginning of this section, which can be represented as $\{ \bx^{\theta} \}_{\theta \in \Theta}$. The only $\bx^{\theta}$ that satisfies Bayes plausibility is $\bx^{\theta} = \mu$. On an intuitive level, this is also easy to see since if the coordination mechanism does not make any recommendation to the receiver, then the receiver has no way to receive any information regarding the state of nature $\omega$ hence any action-independent coordination scheme must be equivalent to revealing no information. The following Table \ref{tab:hardness-results3} summarizes the algorithmics of each mechanism.

\begin{table}[!h]
    \footnotesize
    \centering
   \begin{tabular}{|c|c|c|} \hline
        \makecell{Type-independent Coordination  \\ (Public Signaling)} & Action-independent Coordination &  \makecell{Generic Coordination \\ (Private Signaling) }  \\ \hline
      \makecell{APX-Hard \\ \cite{castiglioni2022bayesian}} &  Trivial  & \makecell{ Polynomial \\ \cite{castiglioni2022bayesian}, $[$Theorem \ref{thm:computability}$]$ }
      \\\hline
    \end{tabular}
    \caption{Complexity landscape for the case of Bayesian   persuasion of a receiver with private type.} 
    \label{tab:hardness-results3}
\end{table}

\vspace{-7mm}
\begin{remark}
Two other problems are closely related to, yet subtly different from, the above problem of information design for a receiver with private type. The first is about selling information to a receiver, as widely studied recently \cite{babaioff2012optimal,bergemann2018design,chen2020selling,yang2022selling}. This turns out to also be a special case of our general model, as elaborated in Appendix \ref{append:sell-info}. Another problem is to persuade multiple receivers from set $\Theta$ with \emph{no externality}, as studied in \cite{arieli2019private,babichenko2017algorithmic,dughmi2017algorithmic}. Our general model can capture their special situation when the sender's utility is additive across receivers, but cannot accommodate  utilities beyond additivity (e..g,  supermodular or submodular sender utilities \cite{arieli2019private}). More detailed discussions on this can be found in Appendix \ref{appendix_sec:related_work}. %
\end{remark}

\subsection{Bayesian Stackelberg Game: A Case with Bounded Strategy Space}
Another special case of the generalized principal-agent problem is the extensively studied model of Bayesian Stackelberg games \cite{stackelberg1934marktform,paruchuri2008playing,conitzer2006computing}. 
Interestingly, unlike the diversity in contract design,  only   the type-independent coordination has been studied in  the Bayesian Stackelberg games, often known as the Bayesian Stackelberg equilibrium. Hence, to paint the entire picture,  we introduce the other two classes    below.     
\begin{itemize}
    \item \emph{Bayesian Stackelberg equilibrium} \cite{paruchuri2008playing,conitzer2006computing}  which corresponds to  type-independent coordination mechanism in succinct form, i.e., a single mixed strategy $\bx^* \in \Delta^d$ that maximizes the leader's expected utility;     
    \item   \emph{Stackelberg Policy}  which corresponds to action-independent coordination mechanism in succinct form, i.e.,  $\{ \bx^{\theta} \} $, which prescribes a mixed strategy for each follower type $\theta$.    
    \item \emph{Stackelberg Mixed Policy} $\{ (\pi(a;\theta), \bx^{a,\theta}) \}_{\theta \in \Theta, a \in \cA} $  
    which corresponds to our general succinct coordination mechanism. Here, the leader commits to the following ``mixed policy'' by first eliciting follower type $\theta$ and then  playing mixed strategy $\bx^{a,\theta}$ with probability $\pi(a;\theta)$, under which the follower type $\theta$ should be incentivized to play $a$ as his best response.  
\end{itemize} 

Since the complexity of the Stackelberg  policy and Stackelberg mixed policy are not studied previously, we resolve their complexity in this paper by proving the following Theorem \ref{thm:Optimal_IC_leader_strategy} (proof deferred to Appendix \ref{append:stackelberg-hard}). This leads to a fully painted algorithmic landscape for Stackelberg games, as shown in  Table \ref{tab:hardness-results2}.    
\begin{theorem}\label{thm:Optimal_IC_leader_strategy} 
For any constant $\epsilon > 0$, it is NP-Hard to compute $\frac{1}{|\Theta|^{1-\epsilon}}$-approximation incentive compatible leader strategy in a Stackelberg game with Bayesian followers. 
\end{theorem}
\begin{small}
\begin{table}[tbh]
    \footnotesize
    \centering
    \begin{tabular}{|c|c|c|c|} \hline
         & Bayesian Stackelberg Equlibrium & Stackelberg Policy  & Stackelberg Mixed Policy \\ \hline
      \multirow{2}{*}{Hidden Follower Type}   & APX-Hard  &  APX-Hard  & Polynomial \\ 
      & \cite{conitzer2006computing} &  [Theorem \ref{thm:Optimal_IC_leader_strategy}] &  [Theorem \ref{thm:computability}]\\\hline
    \end{tabular}
    \caption{Complexity landscape for the case of Bayesian Stackelberg games with hidden follower types.} 
    \label{tab:hardness-results2}
\end{table}
\end{small}

\vspace{-7mm}
\begin{remark} A few remarks are worth mentioning. First,    both the Bayesian Stackelberg equilibrium and Stackelberg policy are NP-hard to be approximated within a factor of $\frac{1}{|\Theta|^{1-\epsilon}}$, whereas the APX-hardness for contract design in Table \ref{tab:hardness-results1} only ruled out constant approximations by \cite{guruganesh2021contracts,castiglioni2022designing}. It is an interesting question to understand whether contract design is fundamentally easier to approximate or previous hardness results could be further strengthened. Second, in the literature, there have been studies that combine Stackelberg games with signaling. For example, \citet{xu2016signaling} studies the case where the leader can commit to a mixed strategy $\mu^\theta$ for each receiver type $\theta$, and meanwhile employs signaling scheme $\pi^{\theta}$ to strategically reveal partial information about the true action $i \in [d]$ that the sender takes. It turns out that this richer model can also be captured by our general principal-agent model where  the coordination mechanism $\pi(\bx, a; \theta): \Delta^d \times \cA \times \Theta \to \RR$ needs to additionally satisfies the supplemental constraints $\sum_{\bx, a} \pi(\bx, a; \theta) \in \Delta^d$ since the summation has to correspond to a leader mixed strategy. 
\end{remark}

\section{Overcoming Adverse Selection via Information Acquisition} \label{sec:info-acq}

An important theme in general principal-agent problems is  the \emph{adverse selection} problem~\cite{akerlof1978market, myerson1982optimal}. While designing the optimal coordination mechanism is one effective way for the principal to overcome adverse selection, in this section we investigate another commonly adopted approach in practice, and also widely studied in the literature, to further mitigate adverse selection --- that is, through acquiring costly information about the agent's type.  In practice, the principal often conducts interviews or offers internships (which have different costs) as a way to acquire signals about the agent's private type of skill level. %
Based on these examination results, the principal can design a more effective coordination mechanism (e.g., in the form of contracts or signals) to work with the agent.  
This motivates us to study the natural algorithmic question arising in this situation: how to acquire information optimally in our general principal-agency model by balancing the cost of information acquisition and the gain from mitigated adverse selection?

\subsection{Information Acquisition in Generalized Principal Agency} 
We model the information acquisition process as a signaling scheme $\tau: \Theta \to \Delta(\Sigma)$, through which  the principal observes a (possibly randomized) signal $\sigma \sim \tau(\theta)$ about the hidden agent type $\theta$. Mathematically, a signaling scheme $\tau: \Theta \to \Delta_{\Sigma}$ is a mapping from the agent's type $\theta$ to a distribution over a set of possible signals $\Sigma$. The scheme $\tau$ can be formally described by $\{\tau(\sigma|\theta) \}_{\theta \in \Theta, \sigma \in \Sigma}$ in which  $\tau(\sigma|\theta)$ is the probability of generating a signal $\sigma$ conditioned on agent type $\theta$. As a real-world example, $\tau$ can be implemented as an interview process in which $\Sigma$ includes the ranks of interview outcomes and $\tau(\sigma|\theta)$ describes how the interview rank could signal underlying agent type.  

With the correlation between $\theta$ and $\sigma$ specified by $\tau$ as well as the agent's type distribution   $f$,  the principal can infer the posterior type distribution with Bayes rule: for any signal $\sigma$, denote the posterior probability of $\theta$ as $\sigma(\theta) := \PP(\theta | \sigma) = \frac{  \tau(\sigma | \theta) f(\theta) }{ \sum_{\theta'} \tau(\sigma | \theta') f(\theta') }$; henceforth, we will directly think of a signal by its induced posterior, i.e., $\sigma \in \Delta(\Theta)$. Let $p = \langle \tau | f \rangle $ be the induced (Blackwell) experiment from prior $f$ and signaling scheme $\tau$, and $p(\sigma) := \tau(\sigma | \theta) f(\theta)$ denotes the probability of $\sigma$ being realized. It is well-known that the design of a signaling scheme is   equivalent to designing a Blackwell experiment $p \in \Delta(\Delta(\Theta))$ that conforms to the Bayes plausibility constraint $\int_{\sigma\in \Delta(\Theta)} p(\sigma) d\sigma = f$ \cite{aumann1995repeated}.  

We assume acquiring information through   experiments (e.g., conducting interviews or offering internships) is costly. Let $c: \Delta(\Delta(\Theta)) \to \RR$ be the cost function for implementing an experiment. We adopt a common assumption that $c$ is \emph{posterior separable} such that $c(p) = \int_{\sigma\in \Delta(\Theta)} p(\sigma) h(\sigma) d\sigma - h(f)$, and $h: \Delta(\Theta) \to \RR$ denotes the cost function associated with each signal (posterior) $\sigma$ and $h(f)$ is the fixed cost for the prior $f$. This property allows us to directly work with the cost function of signal $h$ in the remainder of this section, as the constant cost $h(f)$ does not affect our optimization.  
We also assume that $h$ is convex, followed from the natural notion of Blackwell's informativeness. That is, a more informative signaling scheme (in terms of second-order stochastic dominance) always generates more utility for the decision maker with convex utility and is thus more costly.
A concrete example of such $h$ is the negative entropy, where $h(\sigma)=\sum_{\theta\in \Theta} \sigma(\theta)\ln(\sigma(\theta))$, which is a popular choice to the measure the amount of information in a signal. More such examples can be found in \cite{frankel2019quantifying}. In addition, notice that when $h$ is convex, the total cost $c$ is always non-negative; the problem degenerates when $h$ is linear, as $c(p)$ becomes a constant  regardless of the choice of $p$.

We abstractly represent  the principal's   utility under the posterior type distribution $\sigma$ as a function $u^*: \Delta(\Theta) \to \RR$ --- specifically,   $u^*(\sigma)$ can be set as the maximal objective of CP\eqref{lp:transformed} that captures the principal's expected utility under the optimal coordination mechanism with the agent's type distribution $\sigma$. With a higher level of abstraction, one can also think of $u^*(\sigma)$ as a decision maker's optimal utility under some state distribution $\sigma$, which is a more common setup for information acquisition; we will come back to this point in the later part of this section to illustrate the general applicability of our results. 
Given the prior $f$ and cost function $h$,   the principal's optimal information acquisition scheme $p$ can be formulated as a maximizer of the following optimization program: 
\begin{equation}\label{eq:info-acquisition}
  \bar{u}(f) :=  \max_{p \in \Delta(\Delta(\Theta) )}  \int_{\sigma\in \Delta(\Theta)} p(\sigma) \big[ u^*(\sigma)  - h(\sigma) \big] d\sigma  \ \textup{ subject to } \int_{\sigma\in \Delta(\Theta)} p(\sigma) d\sigma  = f,
\end{equation}
where $\bar{u}: \Delta(\Theta) \to \RR$ denotes the principal's optimal objective under $f$ by optimally acquiring type information. Such $\bar{u}$ is also known as the concavified function of $u^* - h$~\cite{aumann1995repeated,kamenica2011bayesian}. %

However, it remains unclear how to solve for the optimal information acquisition scheme, as OP~\eqref{eq:info-acquisition} involves infinitely many variables and has an objective that may be neither convex nor concave, despite that both $u^*$ and $h$ are convex.  To tackle these challenges, a natural approach is to resort to the revelation principle to simplify the design space, as we did in Section~\ref{sec:opt}. In the remainder of this section, we will demonstrate how this approach turns out to be useful under certain conditions, though the general problem turns out to be provably   intractable.

\subsection{The Tractable Cases: Efficient Concavification via Partition Oracles}

We first present an approach to simplify the optimization problem  that draws inspiration from the recent literature on algorithmic persuasion, which enables us to concavify the target function $u^* - h$ into $\bar{u}$ efficiently.
In particular, we define a notion of partition oracle, based on which we can construct a polynomial time algorithm to solve for the optimal information acquisition scheme.

\begin{definition}[Partition Oracle]
We say the function $u: \cS \to \RR$ admits an  partition oracle $\cP$ if it constructs a partition of the domain $\cS$ of $u$  
$\cP(u)=\{ \cS_i : i\in [n] \}$ such that $\bigcup_{i\in [n]} \cS_i = \cS$ and the function $u$ is concave over convex set $ \cS_i $ for any $i\in [n]$.
\end{definition}

\begin{theorem}[Computing the Optimal Information Acquisition Scheme] \label{thm:computing-optimal-info-scheme}
Given a partition oracle $\cP$ for the utility function $u^*$, the optimal information acquisition scheme under any convex cost function $h$ and prior $f$ can be solved in polynomial time w.r.t. $\abs{ \cP(u^*) }$.  
\end{theorem}
\begin{proof}[Proof Sketch]

The proof follows a similar two-step approach as we did in Section~\ref{sec:opt}. At a high level, we first restrict the design space of the information acquisition scheme through a revelation principle. While the resulting optimization problem \eqref{eq:info-acquisition-simplified} with the bi-linear terms cannot be solved directly, we construct a surrogate convex program \eqref{eq:info-acquisition-simplified-perspective} with the same optimal value. 

Let $\cP(u^*)=\{ \cS_i : i\in [n] \}$ be the output of partition oracle, let $n = |\cP(u^*)|$. We use $\cP(u^*)$ to derive a revelation principle for this design problem: it requires at most one signal $\sigma_i$ for each set $\cS_i$. 
In another word, OP~\eqref{eq:info-acquisition} can simplified as, with variable $ \{ p_i, \sigma_i \}_{i=1}^{n} $,
\begin{equation}\label{eq:info-acquisition-simplified}
    \max_{p_i, \sigma_i} \sum_{i \in [n]} p_i \big[  u^*(\sigma_i)  - h(\sigma_i) \big] \quad \textup{ subject to } \sum_{i \in [n]} p_i \sigma_i  = f, \sigma_i \in \cS_i
\end{equation}
This is easily implied from the property of concavification. For the ``if'' direction, OP \eqref{eq:info-acquisition} achieves no less utility than OP \eqref{eq:info-acquisition-simplified} by relaxing the optimization space. For ``only if'' direction, one can merge all support of $\sigma$ in each one of the partition $\cS_i$ via a convex combination and the concavity ensures that this merge would not decrease the objective.

We now convert OP~\eqref{eq:info-acquisition-simplified} into a concave maximization problem that can be solvable efficiently. 
For this purpose, define the variable $ g_i \coloneqq p_i\sigma_i $ and the perspective functions, $U^*(g_i) = p_i u^*(g_i/p_i) $, $H(g_i) = p_i h(g_i/p_i)$. By substitution, we get OP \eqref{eq:info-acquisition-simplified-perspective} from OP~\eqref{eq:info-acquisition-simplified}, 
\begin{equation}\label{eq:info-acquisition-simplified-perspective}
    \max_{p_i, \sigma_i} \sum_{i \in [n]} \big[  U^*(g_i)  - H(g_i) \big] \quad \textup{ subject to } \sum_{i \in [n]} g_i  = f, \sigma_i \in \cS_i
\end{equation}
We can see that OP \eqref{eq:info-acquisition-simplified-perspective} is a convex program, since $U$ is concave and $H$ is convex w.r.t. $g_i$ by Lemma~\ref{lm:conjugate-property}. 
And after solving OP \eqref{eq:info-acquisition-simplified-perspective}, we can recover the $\sigma_i, p_i$ from each $g_i$ as $p_i = \sum_{\theta} g_i(\theta) $ and $\sigma_i = g_i / p_i $ if $p_i\neq 0$.
We defer the proof of Lemma~\ref{lm:conjugate-property} to Appendix \ref{appendix_section:them:computing-optimal-info-scheme}. It is worth noting that this special kind of  perspective functions, such as the constructed function $H, U^*$, has some additional properties because $h, u^*$ are the convex function that only supports on the simplex. 
    
    \begin{lemma}[Properties of Perspective on Simplex]
    \label{lm:conjugate-property}
    Given any convex function $h: \Delta^d \to \RR$, its perspective function on the simplex, $H: \RR^{d}_{\geq 0} \to \RR$, given by 
    $H(\bx ) = \begin{cases} t\cdot h(\bx /t) & t = \sum_i x_{i} \geq 0 \\ 0 &  \sum_i x_{i} = 0 \\
    \end{cases}$,
    is a convex function that satisfies $H(t \bx) = t H(\bx), \forall t \geq 0.$
    \end{lemma}

\end{proof}

Theorem \ref{thm:computing-optimal-info-scheme} implies several useful corollaries regarding computing the optimal information acquisition scheme.  
Firstly, one can solve the optimal information acquisition scheme in the natural \emph{single-agent} finite-action decision-making problem, which corresponds to a degenerated case of our general model with discrete principal actions but no agent actions (hence the issue of moral hazard becomes absent and the agent's utility function also becomes irrelevant).  
Concretely, in a finite decision problem, the decision maker (the principal) has a utility $u: \cA \times \Theta   \to \RR$ that depends on the state $\theta \in \Theta$ and action $a\in \cA$. Provided that a rational decision maker maximizes her expected utility, we can describe her expected utility under the optimal decision that only depends on the state distribution, $u^*: \Delta(\Theta) \to \RR$. 
When $u$ is concave in $\Theta$, we can observe that $u^*(\sigma) = \max_{a\in \cA} \Ex_{\theta \sim \sigma} u(a ; \theta)$ may be a neither convex or concave function over $\Delta(\Theta)$, making it difficult to directly optimize over. 
Nonetheless, we can easily construct a partition oracle for $u^*$ by solving the best response region for each decision, i.e., $\forall a_i\in \cA, \cS_i = \{\sigma : a_i \in \argmax_{a\in \cA}  \Ex_{\theta \sim \sigma} u(a ; \theta) \}$, and it is clear that, in each $\cS_i$, $u^*(\sigma)=\Ex_{\theta \sim \sigma} u(a_i ; \theta)$ is concave. We thus get the following corollary from Theorem \ref{thm:computing-optimal-info-scheme}.

\begin{corollary}[Optimal Information Acquisition in Finite Decision Problem] \label{coro:decision-problem}
For any finite decision problem $(\cA, \Theta, u)$ with utility $u$ concave in $\Theta$, the optimal information acquisition scheme under any convex cost function $h$ can be solved in polynomial time w.r.t. $|\cA|$.
\end{corollary}

Secondly, our general setup covers the costly persuasion problem~\cite{gentzkow2014costly} and Theorem~\ref{thm:computing-optimal-info-scheme} notably provides an efficient algorithm to compute the optimal costly persuasion scheme.
In a Bayesian persuasion problem~\cite{kamenica2011bayesian} $(\cA, \Theta, u, v)$ with receiver action space $\cA$, state space $\Theta$, sender and receiver utility $u, v: \cA \times \Theta \to \RR$, $u^*(\sigma) = \Ex_{\theta \sim \sigma} u(a^{\sigma} ; \theta)$ is the sender's utility given receiver's best response $a^{\sigma}$ under posterior $\sigma$, i.e., $a^{\sigma} = \max_{a \in \cA} \Ex_{\theta \sim \sigma} v(a ; \theta)$. 
When $u, v$ is linear in $\Theta$ and $\cA$, we can determine a partition for $u^*$ by solving the best response region for each of the agent's action, i.e., $\forall a_i\in \cA, \cS_i = \{\sigma : a_i \in \argmax_{a \in \cA}  \Ex_{\theta \sim \sigma} v(a ; \theta) \}$. We can see that in each $S_i$, $u^*(\sigma)=\Ex_{\theta \sim \sigma} u(a_i ; \theta)$ is linear.
When the sender's design of signal schemes involves a convex cost function $h$, Theorem \ref{thm:computing-optimal-info-scheme} implies the following corollary.\footnote{\citet{gentzkow2014costly} assumed concave cost function, which is seemingly different from yet actually equivalent to our setup, since they define the cost function $c$ in the opposite way. Our definition is more natural from the perspective of the principal (as opposed to the information designer) who only knows the prior and cannot observe the true state.
}

\begin{corollary}[Optimal Signaling in Costly Persuasion]
For any costly persuasion problem $(\cA, \Theta, u, v)$ where $u, v$ is linear in $\Theta$ and $\cA$, the optimal signaling scheme under convex cost function $h$ can be solved in polynomial time w.r.t. $|\cA|$.
\end{corollary}

That said, this result does not generalize to the persuasion problem with multiple receivers, as there are exponentially many combinations of receivers' actions, in which case is indeed NP-hard to find the optimal signaling scheme, even without information cost~\cite{dughmi2017algorithmic}.

\subsection{The Hardness of Optimal Information Acquisition}

As we have identified tractable cases where the partition oracle can be efficiently constructed for concavification, we now show that it is APX-hard to compute the optimal information acquisition scheme in general. It is worth noting that this hardness result is non-trivial and perhaps surprising in the following sense: while both the utility function $u^*$ and cost function $h$ are convex hence easy to maximize in $\Delta(|\Theta|)$ separably --- in fact, we know any convex function is maximized at a vertex of its feasible region, which is just one of the $|\Theta|$ vertices of $\Delta(|\Theta|)$. Nevertheless,  it suddenly becomes intractable for solving the optimal information acquisition which just concavifies their difference $u^*-h$ over the simplex. %

\begin{theorem}[Hardness of Costly Information Acquisition]\label{thm:hard-info-acquisit}
Unless P=NP, for any $\epsilon \in (0,1)$, there is no algorithm in polynomial time of $1/\epsilon$ to compute a  $\frac{1}{|\Theta|^{2-\epsilon}}$-optimal information acquisition scheme for the principal, even when the  cost function $h$ is convex  and when the agent has no actions. 
\end{theorem}
\begin{proof}[Proof Sketch]
The proof involves three major reduction steps: 
First, it is to realize that the decision problem is just a special case of the general principal-agent problem and therefore it suffices to show that solving optimal information acquisition in some single-agent decision problem is APX-hard. 
Next, we show that solving for the optimal information acquisition problem, i.e., concavifying $u^*-h$ is at least as hard as maximizing $u^*-h$. 
Lastly, we present a reduction from the maximum independent set problem to maximizing $u^*-h$ in a class of carefully constructed problem instance.

The first reduction is straightforward, due to the fact that for any decision problem $(\cA, \Theta, u)$, one can construct another principal-agent problem where
the principal and agent share the exact same utility function $u$, and utility is independent of the principal's strategy, but only function of the agent's action space, $\cA$ and the agent's type $\Theta$ (the state space of the decision problem). This means that the set of general principal-agent problems contains the set of the decision problems, and therefore, computing the optimal information acquisition scheme in a generalized principal-agent problem is at least as hard as computing the optimal information acquisition scheme in a single-agent decision problem. 

Next, we prove the second reduction via Lemma~\ref{lm:concavify-hard}.
\begin{lemma}[Concavification is at least as hard as maximization] \label{lm:concavify-hard}
 If a function $\phi: \Delta^k \to \RR$ is NP-Hard to maximize, then $\phi$ must also be NP-Hard to concavify, i.e., there exists some $f \in \Delta(\Theta)$ such that the following optimization problem is NP-Hard to solve  
 $$\bar{\phi}(f) :=  \max_{p_i,\sigma_i}  \sum_{i=1}^{k} p_i \phi(\sigma_i)   \ \textup{ s.t. } \sum_{i=1}^{k} p_i \sigma_i = f. $$
\end{lemma}

Lastly, it suffices to construct a decision problem instance $(\cA, \Theta, u)$ with cost function $h$ such that  $u^* - h $ is APX-Hard to maximize. 
Consider the following decision problem $(\cA, \Theta, u)$ with the state space $|\Theta| = k$ and singleton action space $\cA$. Hence, the function utility $u$ only depends on the state and is equivalent to the utility $u^*$ under optimal action, i.e., $u^*(\sigma)=u(a, \sigma), \forall a, \sigma$. 
At each state distribution $\sigma \in \Delta(\Theta)$, we set the utility function to be, for some $e_{i,j} \in \{0, 1\}$,
\[u^*(\sigma) =  \sum_{i\in [k]} \max  \{\sigma_i - \sum_{j\in [k]} e_{i,j} \sigma_j, 0 \} .\]
In addition, set the convex cost function as $h(\sigma) =  \max\{ \norm{\sigma}_{\infty}-\frac{1}{k}, 0  \}$. We claim that $ u^*(\sigma) - h(\sigma) $ is hard to maximize.
This relies on the following equality
\begin{equation}\label{eq:maximization-equivalence}
  \max_{\sigma \in \Delta^k} u^*(\sigma) - h(\sigma) = \frac{1}{k} \cdot \max_{x \in [0, 1]^k} u^*(x),   
\end{equation}
which implies that any $\epsilon$-optimal solution to $\max_{\sigma \in \Delta^k} u^*(\sigma) - h(\sigma) $ can be turned into an $(\frac{1}{k}\epsilon)$-optimal solution to $\max_{x \in [0, 1]^k} u^*(x)$. By Lemma~\ref{lm:max-ind-set}, there is no algorithm in polynomial time of $1/\epsilon$ to find an $\frac{1}{k^{1-\epsilon}}$-optimal solution to $\max_{x \in [0, 1]^k} u^*(x)$. Hence, there is no algorithm in polynomial time of $1/\epsilon$ to find an $\frac{1}{k^{2-\epsilon}}$-optimal solution to $\max_{\sigma \in \Delta^k} u^*(\sigma) - h(\sigma)$ and the problem is APX-hard.

\begin{lemma}[\citet{NEURIPS2021_e0126439}, Theorem 4]\label{lm:max-ind-set}
Unless P=NP, for any $\epsilon \in (0,1)$, there is no algorithm in polynomial time of $1/\epsilon$ that can solve for $\frac{1}{k^{1-\epsilon}}$-approximation of the following convex maximization problem for some $e_{i,j} \in \{0, 1\}$,
   \begin{equation*}\label{eq:independent_set} \max_{x \in [0, 1]^k} \sum_{i\in [k]} \max  \{x_i - \sum_{j\in [k]} e_{i,j} x_j, 0 \}.
  \end{equation*}
\end{lemma}
We defer the proof of Equation \eqref{eq:maximization-equivalence} and Lemma \ref{lm:concavify-hard} to appendix \ref{appendix_sec_info_hard}.
\end{proof}

\section{Conclusion}

This paper   extends \citet{myerson1982optimal}'s seminal work  to a   generalized principal-agent model that accommodates a rich set of application domains and design considerations. Rather than zeroing in on a specific setup, this paper chooses a less conventional route, zooming out for a panoramic view of the subject matter. From this particular viewpoint, we present a simplified and unified algorithmic solution to efficiently determine the principal's optimal coordination mechanism in the presence of moral hazard and adverse selection. In addition, our comprehensive review of coordination mechanism design in restricted design space highlights the diversity of challenges within the general principal-agent framework and provides a thorough investigation of the algorithmic properties of these special design problems. Finally, on the path of mitigating adverse selection through information acquisition, we delve into the computational intricacies of function concavification, another common problem structure that is shared in various design problems.  

The general framework also opens up several avenues for further investigation. One direction is to understand whether the complexity gaps between different design problems can be closed or whether some design problems are intrinsically harder than others. Another direction is to extend the current framework to the setup with multiple agents and explore how to derive similar generic results under common conditions such as submodular utilities. Lastly, the algorithmics of costly information acquisition can be refined on whether an efficient approximation scheme exists under the entropy cost or other common cost functions.

\bibliographystyle{ACM-Reference-Format}
\bibliography{refer}

\section{Additional Discussion on the Related Work}\label{appendix_sec:related_work}
The problem of information acquisition to reduce adverse selection has been studied in the literature of asset markets, where the buyer often wants to acquire costly information or signals about assets' types from outside sources, also known as \emph{screening}.
\citet{jang2021adverse,wang1998debt,maug1998large,chemla2014skin}. These studies analyze the equilibrium structure and properties of the induced game. For instance, in an important earlier work,    
\citet{maug1998large} showed that large shareholders conduct costly information acquisition
only if the asset market is sufficiently liquid; 
\citet{chemla2014skin} demonstrated that  market liquidity incentivizes lenders to acquire costly information about asset quality. Similar costly information acquisition to reduce adverse selection has been studied in the labor market where the employer  as the principal often acquires signals to screen job applicants (i.e., agents) with private type \cite{sims2003implications,menzio2011efficient}. A more decision-theoretic perspective of costly information acquisition for reducing uncertainty is studied by \citet{matvejka2015rational}.  In contrast, our paper adopts an algorithmic approach to understand the computational complexity of solving for optimal information acquisition with respect to different information cost structures.

The problem of Bayesian persuasion of multiple receivers is studied in \cite{arieli2019private,babichenko2017algorithmic,dughmi2017algorithmic}. The no externality assumption means any receiver's utility is independent of any other receiver's action. In this case, when the sender's utility is additive across different receiver types, the problem will be mathematically equivalent to persuading a single receiver whose type $\theta$ is drawn  from set $\Theta$ uniformly at random. However, this problem starts to be different and more complex when the sender's utility is not additive across the receivers. For instance, \cite{arieli2019private,babichenko2017algorithmic,dughmi2017algorithmic} have considered supermodular or submodular sender utilities, which is the point their problem starts to depart from the problem here. Our framework can be generalized to the situation of single-principal-multiple-agent interactions to accommodate these more general models. That is an interesting future direction, though is out of the scope of the present paper.

\section{Proof of Theorem \ref{thm:revelation-principle}}\label{appendix_sec:thm_rp}
\begin{proof}

First, we generalize Proposition \ref{prop:revelation-principle} to accommodate supplemental constraints. Let us note that Proposition~\ref{prop:revelation-principle} by \citet{myerson1982optimal} applies for the case of infinite action space $\cX$ (as is explicitly pointed out by the author), but it does not consider the supplemental constraints on $\pi$ that we have just introduced in this paper. Hence, we first check that Proposition~\ref{prop:revelation-principle} holds even with supplemental constraints. 

Specifically, for any (non-direct, non-truthful) coordination mechanism $M$, we define $\Sigma$ as the set of all messages which agent might receive, $\cR$ as the set of possible reports that might be sent by the agent. As a result, we have $\rho: \Theta \rightarrow \cR$ as the agent's reporting strategy, $\delta: \Sigma \times \Theta \rightarrow \mathcal{A}$ as the agent's choice of decision, and the coordination policy $\pi: \cX\times\Sigma\times \cR \to [0,1]$ and $\sum_{\bx \in \cX, m\in \Sigma} \pi\big(\bx, m; \rho(\theta)\big) = 1$ for all $\theta \in \Theta$. The supplemental constraint is given by,
$\label{eq:sup_constr_nondirect} \sum_{\bx\in\mathcal{X}} \sum_{m \in \Sigma} \pi(\bx, m, \rho(\theta))\cdot \bx \in \mathcal{C}. $ To construct a direct, truthful mechanism, we can simulate the equilibrium in $M$. That is, we define $\delta^{-1}(a, \theta) = \big\{m|\delta(m, \theta) = a\big\}$ and $\pi^*: \mathcal{X}\times \mathcal{A} \times\Theta \rightarrow [0,1]$ such that
$\pi^*(\bx, a; \theta) = \sum_{m\in \delta^{-1}(a, \theta)} \pi\big(\bx, m, \rho(\theta)\big)$. As a result, the direct and truthful coordination mechanism satisfies the additional constraints as,
\begin{equation*}
\begin{split}
    \label{eq:sup_constr_direct}
    \sum_{\bx\in\mathcal{X}} \sum_{a \in \mathcal{A}} \pi^*(\bx, a; \theta)\cdot \bx  
   = \sum_{\bx\in\mathcal{X}} \sum_{a \in \mathcal{A}} \,\sum_{m\in \delta^{-1}(a, \theta)} \pi\big(\bx, m, \rho(\theta)\big) \cdot \bx 
   = \sum_{\bx\in\mathcal{X}} \sum_{m \in \Sigma} \pi(\bx, m, \rho(\theta))\cdot \bx \in \mathcal{C}.
\end{split}
\end{equation*} 

It now remains to show that the support of $\pi(\cdot, a; \theta)$ needs not to be larger than $1$. 
  To prove this, let us pick an arbitrary coordination mechanism $\pi$ that is truthful and direct. 
We construct a new coordination mechanism $\bar{\pi}$ as follows.
For every $a\in \cA$ and $\theta \in \Theta$, let the $\cX^{a, \theta}$ be the set of principal strategies against which action $a$ is a best response of a type-$\theta$ agent. 
Moreover, define the following strategy
$\bar{\bx}^{a, \theta} \coloneqq \frac{\sum_{\bx \in \cX^{a, \theta}} \pi(\bx, a;\theta) \cdot \bx }{\sum_{\bx \in \cX^{a, \theta}}  \pi(\bx, a;\theta)}$,
and let the new coordination mechanism choose this strategy with probability $\bar{\pi}(a;\theta) = \sum_{\bx \in \cX^{a, \theta}} \pi(\bx, a;\theta)$.\footnote{Throughout this proof, if $\cX^{a,\theta}$ is an infinite set we can replace the summation operation to integration.
}

It is easy to see that $\bar{\bx}^{a, \theta} \in \cX$ as it is a convex combination of the strategies in $\cX^{a, \theta}$. Moreover, for every $\theta\in \Theta$, we have $\sum_{a\in \cA} \bar{\pi}(a; \theta) = 1$, so $\bar{\pi}(\cdot;\theta)$ is a valid probability distribution. 
We next show the following facts about $\bar{\pi}$ to complete the proof.
\begin{enumerate}
   \item 
    First, $\bar{\pi}$ remains feasible. In particular, it satisfies the supplemental constraints defined in \eqref{eq:supplimental-const}: 
    for all $\theta \in \Theta$, we have
    \begin{align*}
\sum_{a \in \cA} \bar{\pi}(a;\theta) \cdot \bar{\bx}^{a, \theta} 
    = \sum_{a \in \cA} \sum_{\bx \in \cX^{a,\theta}} {\pi}(\bx, a;\theta) \cdot \bx 
    = \sum_{a \in \cA} \sum_{\bx \in \cX} {\pi}(\bx, a;\theta) \cdot \bx 
    \in \cC_{\theta}.
    \end{align*}

    \item 
    Second, $\bar{\pi}$ is incentive compatible:
    for any $\theta' \in \Theta$, we have 
\begin{align*}
 & \quad \sum_{a \in \cA} \bar{\pi}(a;\theta) \cdot V(\bar{\bx}^{a, \theta}, a; \theta)  
=  \sum_{a \in \cA} \sum_{\bx \in \cX^{a, \theta}} \pi(\bx, a;\theta) \cdot V(\bx, a; \theta)   
=  \sum_{\bx \in \cX} \sum_{ a \in \cA} \pi(\bx, a;\theta) \cdot V(\bx, a; \theta)   \\
&\quad \geq \sum_{\bx \in \cX} \sum_{a \in \cA} \big[ \pi(\bx, {a}; \theta') \cdot \max_{a'\in \cA} V(\bx, a'; \theta ) \big] 
= \sum_{a \in \cA} \sum_{\bx \in \cX^{a, \theta}} \big[ \pi(\bx, a;\theta') \cdot \max_{a'\in \cA} V(\bx, a'; \theta ) \big]  \\
 &\quad \geq  \sum_{a \in \cA} \bar{\pi}(a;\theta') \cdot \max_{a'\in \cA} V(\bar{\bx}^{a, \theta}, a'; \theta ) \big].
\end{align*}
Here the first transition follows by the definition of $\bar{\bx}^{a, \theta}$ and the linearity of $V$ with respect to $\bx$. %
The second follows by merging the strategies $\bx \in \cX^{a, \theta}$ into the summation terms. The third follows by the truthfulness constraint of the original coordination mechanism. The fourth follows by separating the strategies $\bx \in \cX^{a, \theta} $ from the summation terms. The last transition uses Jensen's inequality and the convexity of point-wise maximization.

    \item 
    Finally, 
    merging the strategies in $\cX^{a, \theta}$ will not cause the principal's utility to decrease.
    Specifically, for all $a \in \cA$ and $\theta \in \Theta$, using Jensen's inequality and the concavity of $U$ with respect to $\bx$, we get that
\begin{align*}
 \bar{\pi}( a; \theta) \cdot U(\bar{\bx}^{a, \theta}, a; \theta) 
\geq \sum_{\bx\in \cX^{a, \theta}} \pi(\bx, a; \theta) \cdot U(\bx, a; \theta) 
= \sum_{\bx\in \cX} \pi(\bx, a; \theta) \cdot U(\bx, a; \theta),
\end{align*}
where the equality is due to the incentive compatibility of $\bar{\pi}$ (i.e., $ \pi(\bx, a; \theta) = 0$ for $\bx \notin \cX^{a, \theta}$ ).
Now that $\bar{\pi}$ is also incentive compatible as we argued above, the principal's overall expected utility can be written as $\sum_{\theta \in \Theta} f(\theta) \sum_{a \in \cA} \bar{\pi}(a; \theta) \cdot U(\bar{\bx}^{a, \theta}, a; \theta)$. Hence, 
\begin{align*}
\sum_{\theta \in \Theta} f(\theta) \sum_{a \in \cA} \bar{\pi}( a; \theta) \cdot U(\bar{\bx}^{a, \theta}, a; \theta)
\ge \sum_{\theta \in \Theta} f(\theta) \sum_{a \in \cA} \sum_{\bx\in \cX} \pi(\bx, a; \theta) \cdot U(\bx, a; \theta),
\end{align*}
which means that the principal's utility for $\bar{\pi}$ is at least as large as that for $\pi$.
\qedhere
\end{enumerate}
\end{proof}

\section{Omitted Proof of Theorem \ref{thm:computability}}\label{appendix_sec:thm_computability}

\subsection{ Proof of Lemma \ref{lm:convex-z}}
\begin{proof}
To prove that $\cP$ is a convex set, pick any two element $( \lambda_1, \lambda_1 \bx_1), ( \lambda_2, \lambda_2 \bx_2)$ and we show that for any $\alpha\in [0,1]$, the element as a convex combination $( \alpha\lambda_1 + (1-\alpha)\lambda_2, \alpha \lambda_1 \bx_1 + (1-\alpha)\lambda_2 \bx_2)$ is still in  $\cP$. 

It suffices to consider two cases and show in both cases that the constructed element can be written as $(\lambda_3, \lambda_3 \bx_3)$ where $\lambda_3 \in [0,1] $ and $\lambda_3 \in \cX$.

\begin{enumerate}
\item 
     Any of $\lambda_1, \lambda_2$ is $0$: Without loss of generality, suppose $\lambda_1 = 0$, then we have $\lambda_1 \bx_1 = \zero$ and the new element must be in the form of $((1-\alpha)\lambda_2,  (1-\alpha)\lambda_2 \bx_2)$. Then, we can let $\lambda_3 = (1-\alpha)\lambda_2 \in [0, 1]$ and $\bx_3 = \bx_2 \in \cX$. 

\item 
None of $\lambda_1, \lambda_2$ is $0$:
First, we can easily see that $ \lambda_3 = \alpha\lambda_1 + (1-\alpha)\lambda_2 \in [0,1]$. 

Second, 
we have $\bx_3 = \frac{\alpha \lambda_1 \bx_1 + (1-\alpha)\lambda_2 \bx_2}{\alpha\lambda_1 + (1-\alpha)\lambda_2} =  \frac{\alpha \lambda_1 }{\alpha\lambda_1 + (1-\alpha)\lambda_2}\bx_1 + \frac{(1-\alpha) \lambda_2 }{\alpha\lambda_1 + (1-\alpha)\lambda_2}\bx_2 \in \cX$, since $ \frac{\alpha \lambda_1 }{\alpha\lambda_1 + (1-\alpha)\lambda_2} \in [0, 1]$ and $ \frac{\alpha \lambda_1 }{\alpha\lambda_1 + (1-\alpha)\lambda_2} + \frac{(1-\alpha) \lambda_2 }{\alpha\lambda_1 + (1-\alpha)\lambda_2} = 1$. 
\qedhere
\end{enumerate}
\end{proof}

\subsection{Proof of Lemma \ref{lm:lp-equivalence}}
\begin{proof}
For $\sup (\textup{CP} \eqref{lp:transformed}) = \sup (\textup{OP} \eqref{lp:original})$, we need to show both directions holds:

\begin{enumerate}
    \item $\sup (\textup{CP} \eqref{lp:transformed}) \geq \sup (\textup{OP} \eqref{lp:original})$.
    That is, for any $ \{( \pi(a; \theta), \bx^{a, \theta}) : \theta \in \Theta, a \in \cA\}$ in OP \eqref{lp:original}, we can construct the mechanism $ \{( \pi(a; \theta), \bx^{a, \theta}) : \theta \in \Theta, a \in \cA\}$ in CP \eqref{lp:transformed} that achieves the same utility. The construction is done by setting $\bz^{a, \theta} =  \pi(a; \theta) \cdot \bx^{a, \theta}$ and keeping $\pi(a; \theta)$ the same as before 
    for each $(a, \theta)$. It is immediately clear according to this construction that $ \{( \pi(a; \theta), \bx^{a, \theta}) : \theta \in \Theta, a \in \cA\}$ satisfies all constraints in CP \eqref{lp:transformed} and has exactly the same objective as the original mechanism.
    
    \item $\sup (\textup{CP} \eqref{lp:transformed}) \leq \sup (\textup{OP} \eqref{lp:original})$. 
    That is, for any $ \{( \pi(a; \theta), \bz^{a, \theta}) : \theta \in \Theta, a \in \cA\}$ in CP \eqref{lp:transformed}, we can construct a mechanism $ \{( \pi(a; \theta), \bx^{a, \theta}) : \theta \in \Theta, a \in \cA\}$ in OP \eqref{lp:original} that achieves the same utility. The construction is done by setting $\bx^{a, \theta} = \begin{cases}\bz^{a, \theta} / \pi(a; \theta) & \pi(a; \theta) \neq 0 \\ 0 & \pi(a; \theta) = 0 \end{cases}$ and keeping $\pi(a; \theta)$ the same as before for each $(a, \theta)$. %
    Observe that for any $(a, \theta)$, if $(\pi(a; \theta), \bz^{a, \theta})=(0, \zero)$, both the principal's and agent's utilities for action $a$ and type $\theta$ must be zero in CP~\eqref{lp:transformed}; Similarly, both utilities are zero in OP \eqref{lp:original} when $(\pi(a; \theta), \bx^{a, \theta})= (0, \zero)$ .
    It remains to verify that $ \{( \pi(a; \theta), \bx^{a, \theta}) : \theta \in \Theta, a \in \cA\}$ satisfies all the constraints in OP~\eqref{lp:original} and has the exact same objective as the original mechanism, which is straightforward.
    \qedhere
\end{enumerate}
\end{proof}

\subsection{Proof of Lemma \ref{lm:compact-z}}
\begin{proof}

Clearly, $\cP' = \{(\lambda, \bx) : \bx \in \cX, \lambda \in [0,1]\} = [0,1] \times \cX$ is a compact set.
The set $\cP$ is the image of $\cP'$ under the continuous mapping $\cT(\lambda, \bx) = (\lambda, \lambda\bx)$.
Since compactness is preserved under continuous mapping, it follows immediately that $\cP$ is compact. 
\end{proof}

\subsection{Proof of Lemma \ref{lm:unbounded-z}}
\begin{proof}

Let $\hat{\cP} = \{ ( 0, \bz) : \bz \neq \zero \}$. We will argue that for all $p \in \bar{\cP}$, it holds that $p \in \cP$ if and only if $p \notin \hat{\cP}$.
This easily generalizes to the stated result in the lemma as CPs~\eqref{lp:transformed-closure} and \eqref{lp:transformed} differ only in the second sets of constraints.

Pick arbitrary $p = (\lambda, \bz) \in \bar{\cP}$.
Since $p \in \bar{\cP}$, by definition, there exists a convergent sequence $p^1, p^2, \dots$, such that $p^{(\ell)} \in \cP$ for all $\ell = 1,2,\dots$, and $\lim_{\ell \to \infty} p^{(\ell)} = p$.
Let $p^{(\ell)} = (\lambda^{(\ell)}, \bz^{(\ell)})$.
Hence, $\lim_{\ell \to \infty} \lambda^{(\ell)} = \lambda$ and $\lim_{\ell \to \infty} \bz^{(\ell)} = \bz$.
The fact that $p^{(\ell)} \in \cP$ implies that $\lambda^{(\ell)} \in [0,1]$, so it must be that $\lambda \in [0,1]$.
Consider the following two cases:

\begin{itemize}
\item If $\lambda = 0$, then either $\bz = \zero$, in which case $p = (0, \zero) \in \cP$, or $\bz \neq \zero$, in which case $p \in \hat{\cP}$. 
Hence, $p \in \cP$ if and only if $p \notin \hat{\cP}$.

\item If $\lambda \in (0,1]$, we have 
$\lim_{\ell \to \infty} \bz^{(\ell)} / \lambda^{(\ell)} = \bz/\lambda$.
The fact that $p^{(\ell)} \in \cP$ implies that $\bz^{(\ell)} / \lambda^{(\ell)} \in \cX$ according to the definition of $\cP$. 
Recall that $\cX$ is a closed set, which means $\bz/\lambda \in \cX$. As a result, $p = (\lambda, \bz) \in \cP$. Moreover, since $\lambda \neq 0$, we also have $p \notin \bar{\cP}$ in this case.  
\end{itemize}
Hence, $p \in \cP$ if and only if $p \notin \hat{\cP}$.
\end{proof}

\section{Additional Special Cases of  Generalized Principal-Agency}\label{appendix:additional_special}
\subsection{Bayesian Contract Design: A Case with Unbounded Strategy Space}
In this model, each agent's action $a$  will emit  a random outcome $i$  drawn from $\{ 1, 2, \cdots, d \} = [d]$, which induce principal reward $r_i\in \RR$.  Let $\br \in \RR^d$ denote the reward vector over   outcomes. An agent of type $\theta$ suffers cost $c^{\theta}_a$ for taking action $a$, whereas action $a$ will lead to outcome $i$ with probability $P^{\theta}_{a,i}$. Let $P^{\theta}_a \in \Delta^d$ denote the probability distribution.  Similar to our general setup, the agent's type $\theta$ is assumed to be drawn from the distribution $f(\theta)$, game parameters are publicly known whereas the agent privately knows $a, \theta$ which the principal does not observe. Notably, the realized outcome $i$ is publicly observable.  The principal's goal here is to design \emph{contracts} that incentivize the  agent  to take an action that (hopefully) has low payment but high return to the principal. Naturally, such an ideal action is generally different for different agent types. Formally, a contract $\bx \in \RR_+^d$ is a vector that describes payment $x_i$ for realized outcome $i$. Naturally, the principal could use different contracts for different agent types, or even randomized contracts.

The Bayesian contract design problem above turns out to be a special case of our general model, with the principal's action space as $\cX = \RR_+^d$ that contains all possible contracts  and is an unbounded set. The principal's and agent's utilities $U, V: \RR^{d} \times A \times \Theta \to \RR$ turn out to be the following affine functions: %
\begin{eqnarray*}
U(\bx, a; \theta) = P^\theta_a \cdot  (\br - \bx)  \qquad  V(\bx, a;\theta) = P^\theta_{a} \cdot  \bx - c^\theta_a 
\end{eqnarray*}

\subsection{Bayesian Persuasion: A Case with supplemental Constraints}
\label{append_sec:info-design}
In this model, each agent type $\theta$ is called a \emph{receiver} who chooses an  action $a \in \cA$ whereas the principal is the \emph{sender}. Both the sender and each receiver's utility depends on the receiver's action $a \in \cA$ and a random state of nature $\omega \in \Omega$ with $|\Omega| = d$. We denote $u^\theta(\omega, a) \in \mathbb{R} $ (resp.  $v^\theta(\omega, a) \in \mathbb{R}$) as the sender's utility  (resp. receiver) when receiver type $\theta$ takes action $a$ under state $\omega$. Let $u^\theta, v^\theta \in \RR^{\Omega \times \cA}$ denote the corresponding utility matrices. The state $\omega$ is assumed to be drawn from some prior distribution $\mu$. All these game parameters are public knowledge. Each receiver privately knows their own type $\theta$ (but the receiver does not observe the realized state $\omega$).  

The persuasion problem assumes that the sender can design \emph{signaling schemes} that strategically reveal partial   information about the realized state $\omega$ hence influence the receiver's action and further the sender's own utility.  A signaling scheme $\pi(s;\omega)$  prescribes the probability of sending some \emph{signal} $s \in \cS$ conditioned on realized state $\omega$. Hence naturally we required $\sum_{s \in \cS} \pi(s;\omega) = 1$ for every $\omega$. Upon receiving any signal $s$, the receiver can infer a posterior distribution of the underlying state $\omega$ (which is all that matters to the receiver's decision making) based on the standard Bayes rule: $$x_{\omega}^s = \Pr(\omega|s) = \frac{ \pi(s; \omega) \times \mu(\omega) }{ \sum_{\omega'}\pi(s; \omega') \times \mu(\omega') }.$$
Hence one could  view each signal $s$ mathematically equivalently as a a distribution over states $\omega$, described by the vector $\bx^s \in \Delta(\Omega)$, which happens with probability $\pi(s) = \sum_{\omega'}\pi(s; \omega') \times \mu(\omega')$. Observe that $\sum_{s \in S} \pi(s) \cdot \bx^s = \mu$. This is not   a coincidence --- the well-known Bayes plausibility result (e.g., \cite{aumann1995repeated}) asserts that any signaling scheme is mathematically equivalent to a distribution $\pi(s)$ over posterior beliefs $\bx^s \in \Delta(\Omega)$ with an additional supplemental constraint $\sum_{s \in S} \pi(s) \cdot \bx^s = \mu$, and vice versa.  
 
We are now ready to observe that the persuasion problem above is a special case of our general principal agent problem. Specifically, the principal's strategy space $\cX = \Delta(\Omega)$ and a generic coordination mechanism is a distribution $\pi(\bx, a; \theta)$, subject to the supplemental constraints due to Bayes plausibility requirements: $\sum_{\bx \in \cX, a} \pi(\bx, a; \theta) \cdot \bx^s = \mu$ for every $\theta$. The sender's and receiver's utility functions $U, V:\Delta(\Omega) \times \cA \times \Theta \to \RR $  are both affine as follows :
\begin{equation}
     U(\bx, a; \theta) = \bx^\top \cdot [u^\theta(\cdot, a)] \qquad V(\bx, a; \theta) = \bx^\top \cdot  [v^\theta(\cdot, a)]  
\end{equation}

\subsection{Selling Information to a Bayesian Decision Maker}\label{append:sell-info}
The problem of selling information, as widely studied recently \cite{babaioff2012optimal,bergemann2018design,chen2020selling,yang2022selling}, is related to yet different from information design (i.e., Bayesian persuasion)  and can also be captured as a special case of our general model. Here the principal is a data broker, who can design signaling schemes to reveal partial information about a random state of nature $\omega \in \Omega$. The information buyer plays the role of the receiver  as in the Bayesian persuasion model above with utility function $v^{\theta}$. Unlike the sender in persuasion,   the broker here does not have any utility function, but instead would like to maximize the charge from the buyer. Previous works by \citet{babaioff2012optimal,chen2020selling} show  that the optimal mechanism here can without loss of generality have the following form: charging a payment $t^\theta_a$ from   buyer o type $\theta$ and meanwhile making obedient recommending action $a$ based on signaling scheme $\pi(a;\theta)$. It is easy to see that this corresponds to a  general principal-agent problem that is similar to that of Bayesian persuasion above, but with different principal strategy and player utilities:
\begin{equation*}
     U([\bx,t], a; \theta) = t \qquad V([\bx,t], a; \theta) = \bx^\top \cdot  [v^\theta(\cdot, a)]  - t
\end{equation*} 
where $t \in \RR$ is the charge to the information buyer of type $\theta$ when he is recommended action $a$, whereas $\bx$ is the agent's posterior belief given action $a$ which is additionally required to satisfy the Bayes plausibility constraint.

\subsection{Bayesian Stackelberg Game: A Case with Bounded Strategy Space}
In this model, the principal is the leader with commitment power whereas the agent is a  follower.  The follower has action set  $\cA = \{1, \cdots, n\}$  whereas the leader has an action set $[d] = \{1, \cdots, d \}$. The leader's strategy is a mixed strategy, i.e., a distribution $\bx \in \Delta^d$ over $[d]$ where   $ \Delta^d = \{\bx \in \RR_+^d: \sum_{i \in [d]} x_i = 1\}$ is the $d$-dimensional simplex.  
The follower has a private type $\theta$ sampled from a probability distribution $f$ supported  on set $\Theta$. When the leader takes action $i \in [d]$ and a follower type $\theta$ takes action $a \in \cA$, the leader receives utility  $u_{i,a}^\theta$ whereas the receiver receives $v_{i,a}^\theta$. Let $u^\theta, v^{\theta} \in \mathbb{R}^{d \times n}$ denotes the corresponding payoff matrices. 

It is easy to see that the above model is a special case of our general principal-agent model where the principal's strategy space $\cX = \Delta^d$ is bounded. The leader's and follower's  utility function as $U, V: \Delta^d \times \cA \times \Theta \to \RR $ can be written as follows:
$$U(\bx, a; \theta) = \bx^\top \cdot [ u^\theta(\cdot, a)]    \qquad V(\bx, a; \theta) = \bx^\top \cdot [ v^\theta(\cdot, a)].   $$  

\section{Proof of Theorem \ref{thm:Optimal_IC_leader_strategy} } \label{append:stackelberg-hard}
We show a reduction from the Max Independent Set (MIS) problem. It is known that no polynomial-time $\frac{1}{|\Theta|^{1-\epsilon}}$-approximation algorithm exists for MIS, unless P=NP. An MIS instance is given by a graph $G = (V, E)$, we construct a game with $K = |V| = |\Theta|$ follower types $\Theta = \{\theta_v:v \in V\}$; each follower $\theta_v$ corresponds to a node $v$ and appears with probability $\frac{1}{k}$. The leader has $2|V|$ actions $\{a_v:v\in V\}\cup\{b_v:v\in V\}$; the follower has $3$ actions $\{1^F, 2^F, 3^F\}$. The follower's utility is given by Table~\ref{table:follower_hardness}. 

\begin{table}[h]
    \centering
    \begin{tabular}{|c|c|c|c|} \hline
      Type $\theta_v$  &  $1^F$ & $2^F$ & $3^F$ \\ \hline
      $a_v$   & 0.1 & 0.1 & 0.1   \\ \hline
      $b_v$   &  & 1 & 1  \\ \hline
      $\{a_{v'}:v'\in \mathcal{N}(v)\}$   & 0.5 &  & 1   \\ \hline
      otherwise & & &0.1\\ \hline
    \end{tabular} 
    \caption{Utility information of the follower with type $\theta_v$. $\mathcal{N}(v)$ denotes the neighbour nodes of $v$ in the graph $G$. All empty elements are 0.}
    \label{table:follower_hardness}
\end{table}
Irrespective of the strategy the leader plays, the leader always gets utility $1$ when the follower responds with action $1^F$, and utility $0$ when the follower responds with $2^F$ or $3^F$. 

Given an independent set $V^*$ of size $k$, consider the following leader mechanism:
\[
\bvec{x}^*(\theta^v) = \begin{cases}
a_v, \quad \text{if } v \in V^* \\
b_v, \quad \text{otherwise}
\end{cases}
\]
We claim the above leader mechanism $\bvec{x}^*(\theta^v)$ is an incentive compatible leader mechanism and gives the expect leader utility $\frac{k}{K}$ because: (1) for any follower type $\theta^v \in V^*$, misreporting $\theta^{\bar{v}}$  leads to the leader playing strategy $a_{\bar{v}}$ if $\bar{v} \in V^*$ (which means $\bar{v} \notin \mathcal{N}(v)$ as $V^*$ is an independent set), or strategy $b_{\bar{v}}$ if $\bar{v} \notin V^*$. Neither strategy makes the follower better off than reporting truthfully. (2) for any follower type $\theta^v \notin V^*$, the follower can already get the maximum possible utility $1$ if they report truthfully. So no follower type has any incentive to misreport in this case. As a result, among the $K$ follower types, the $k$ types in the independent set will respond with $1^F$ which offers the leader a utility of $\frac{k}{K}$.

On the other hand, given any incentive compatible mechanism, we claim that the follower types who respond with $1^F$ must form an MIS. For the sake of contradiction, assume $(v_1, v_2) \in E$ and both follower types $\theta^{v_1}$ and $\theta^{v_2}$ best respond with $1^F$. Note that $1^F$ is strictly dominated by $3^F$ for the follower unless the leader plays $a_v$. Therefore, in order to make follower type $\theta^{v_1} (\theta^{v_2})$ best respond with $1^F$, we must have $\bvec{x}(\theta^{v_1}) = a_{v_1} (\bvec{x}(\theta^{v_2}) = a_{v_2})$. However, note that $(v_1, v_2) \in E$. Consequently,  the follower of type $\theta^{v_1} (\theta^{v_2})$ is incentivized to misreport  $\theta^{v_2} (\theta^{v_1})$, which gives them a higher utility, contradicting the incentive compatibility.

As a result, under the optimal incentive compatible leader strategy, the number of follower types who respond with $1^F$ equals to the size of the MIS. In addition, the number of follower types who respond with $1^F$ is also proportional to the leader's overall utility since the leader only gets utility when the follower responds with $1^F$. Therefore, any $\frac{1}{|\Theta|^{1-\epsilon}}$-approximation algorithm for the optimal incentive compatible leader strategy would also give $\frac{1}{|\Theta|^{1-\epsilon}}$-approximation to the MIS problem.

\section{Omitted Proof of Theorem \ref{thm:computing-optimal-info-scheme}}\label{appendix_section:them:computing-optimal-info-scheme}
\subsection{Proof of Lemma \ref{lm:conjugate-property}}
\begin{proof}
    We first prove the property that $H(t \bx) = t H(\bx), \forall t \geq 0$. When $t=0$, we have $H(t\bx) = 0 = tH(\bx)$. So instead pick any $t>0$, we have 
    \begin{eqnarray*}
        H(t \bx) & =&  \sum_{i} t x_{i} f( \frac{t \bx }{ \sum_i t x_i } )  \\
        & =& t \sum_{i} x_{i} f( \frac{\bx }{ \sum_i x_i } ) \\
        & = & t H(\bx) 
    \end{eqnarray*}
    
    We now prove its convexity by definition. Pick any $\bx, \bx' \in \RR_{\geq 0}^d$, we show that $H( \alpha \bx + (1-\alpha) \bx' ) \leq \alpha H(  \bx) + (1-\alpha) H( \bx' )$. If $\sum_i x_i = 0$ or $\sum_i x'_i = 0$, this is trivially true using the first property. Otherwise $\sum_i x_i > 0$, $H$ is convex as the perspective function of the convex function with domain restricted to the simplex $\Delta^d$.

    \end{proof}

\section{Omitted Proof of Theorem \ref{thm:hard-info-acquisit}}\label{appendix_sec_info_hard}
\subsection{Proof of Lemma \ref{lm:concavify-hard}}
\begin{proof}
We prove the contrapositive statement: if there exists a poly-time algorithm to solve $\bar{\phi}(f)$ for any $f$, then there exists a poly-time algorithm to find the global maximum for $\phi$.

Observe that $\bar{\phi}$ is a concave function with its maximum also being the global maximum of $\phi$. If there exists a poly-time algorithm to solve $\bar{\phi}(f)$ for any $f$, we can evaluate $\bar{\phi}$ at any $f$ in poly-time, which implies that the maximum of $\bar{\phi}$ can be solved in poly-time.
\end{proof}

\subsection{Proof of Equation \eqref{eq:maximization-equivalence}}
\begin{proof}
To prove Equation \eqref{eq:maximization-equivalence} that
\begin{equation*}%
  \max_{\sigma \in \Delta^k} u^*(\sigma) - h(\sigma) = \frac{1}{k} \cdot \max_{x \in [0, 1]^k} u^*(x),   
\end{equation*}
we consider two cases by splitting $\Delta^k$ into two part $\cS = [0, 1/k]^k$ and $\bar{\cS} = \Delta^k \setminus [0, 1/k]^k $: 
\begin{itemize}
    \item For any $\sigma \in \cS$, we have $h(\sigma)=0$. Since $ u^*(t\cdot x) = t\cdot u^*(x), \forall t > 0 $, we have 
    $$ 
    \max_{\sigma \in \cS} u^*(\sigma) - h(\sigma) = \max_{\sigma \in [0, 1/k]^k} u^*(\sigma) = \frac{1}{k}\cdot \max_{x \in [0, 1]^k} u^*(x).
    $$
    \item For any $\sigma \in \bar{S}$, we have $h(\sigma)=\norm{\sigma}_{\infty}-1/k$. We claim that 
    $$
    (*)=\max_{\sigma \in \cS} u^*(\sigma) - h(\sigma) \geq \max_{\sigma \in \bar{\cS}} u^*(\sigma) - h(\sigma).
    $$
    This can be proved by contradiction, suppose there is $\sigma' \in \bar{\cS}$ with $u^*(\sigma') - h(\sigma') > (*) $. Let $ t = \norm{\sigma'}_{\infty}$ and we know by definition $ t \geq 1/k$. As is shown above, for any $t > 0$, $\max_{\sigma \in [0, t]^k} u^*(\sigma) = t\cdot \max_{x \in [0, 1]^k} u^*(x) $, so we have 
    $$ u^*(\sigma') 
    =  \max_{\sigma \in \Delta^k, \norm{\sigma}_{\infty}=t} u^*(\sigma)
    \leq \max_{\sigma \in [0, t]^k} u^*(\sigma) 
    = t \cdot \max_{x \in [0, 1]^k} u^*(x).
    $$
    Applying the definition that $h(\sigma') = t - 1/k$, we get
$$ u^*(\sigma') - h(\sigma') \leq t\cdot \max_{x \in [0, 1]^k} u^*(x) - t + 1/k. $$
    Since $\max_{x \in [0, 1]^k} u^*(x) \leq 1$, we have 
 $$ t\cdot \max_{x \in [0, 1]^k} u^*(x) - t + 1/k 
 \leq \frac{1}{k}\cdot \max_{x \in [0, 1]^k} u^*(x) =  \max_{\sigma \in \cS} u^*(\sigma) - h(\sigma).$$
 Combining the two inequalities above, we get 
 $$
 u^*(\sigma') - h(\sigma') \leq \max_{\sigma \in \cS} u^*(\sigma) - h(\sigma).
 $$
  Hence, a contradiction is reached. 
\end{itemize}
  Combining the analysis on $\cS$ and $\bar{\cS}$, we get the Equation \eqref{eq:maximization-equivalence}.
  $$ 
    \max_{\sigma \in \Delta^k} u^*(\sigma) - h(\sigma) = \max_{\sigma \in \cS} u^*(\sigma) - h(\sigma) = \frac{1}{k}\cdot \max_{x \in [0, 1]^k} u^*(x).
    $$
\end{proof}

\end{document}